%% file: ecma_submission.tex
\documentclass[11pt]{amsart}
\usepackage{amsmath,amsfonts,amssymb,amsthm,epsfig,color, graphicx, natbib, appendix,comment,subcaption, graphicx, amssymb,xcolor,geometry,bm}

\usepackage{setspace} %

\usepackage[colorlinks=true,
    linkcolor=blue,
    filecolor=magenta,      
    urlcolor=blue]{hyperref}
\usepackage{tikz}
\usetikzlibrary{shapes,arrows}
\usetikzlibrary{intersections}
\usetikzlibrary{external}
\usepackage[capitalise,noabbrev,nameinlink]{cleveref}

\DeclareMathOperator*{\diag}{diag}
\newtheorem{theorem}{Theorem}
\newtheorem{lemma}{Lemma}
\newtheorem*{definition}{Definition}

\newtheorem{proposition}{Proposition}
\newtheorem{corollary}{Corollary}
\theoremstyle{definition}
\newtheorem{example}{Example}
\theoremstyle{remark}
\newtheorem{remark}{Remark}

\newcommand{\continueexample}[1]{%
  \begin{example}[\textbf{continued}]
  #1
  \end{example}
}

\DeclareMathOperator*{\argmax}{arg\,max}
\DeclareMathOperator*{\argmin}{arg\,min}

\usepackage{tikz}
\usetikzlibrary{positioning,quotes}

\definecolor{darkblue}{rgb}{0.0, 0.0, 0.55}
\hypersetup{colorlinks,linkcolor={darkblue},citecolor={darkblue},urlcolor={darkblue}} 
\allowdisplaybreaks
\title{Equity pay in networked teams}

\author{Krishna Dasaratha \and Benjamin Golub \and Anant Shah}

\thanks{Dasaratha: krishnadasaratha@gmail.com, Boston University. Golub: ben.golub@gmail.com, Northwestern University. Shah: anantshah2026@u.northwestern.edu, Northwestern University. We thank Vitalii Tubdenov for excellent research assistance. We are grateful (in random order) to Ilya Segal, George Mailath, Aravindan Vijayaraghavan, Omer Tamuz, Marina Halac, Thomas Steinke, Marzena Rostek, Juan Ortner, Alex Wolitzky, Jason Hartline, Evan Sadler, Adam Szeidl, Michael Ostrovsky, Melika Liporace, and Michael Powell, as well as many seminar and conference participants, for valuable conversations and comments.}

\input{notation.tex}

\begin{document}

\begin{abstract}
A group of agents each exert effort to produce a joint output, with the complementarities between their efforts represented by a (weighted) network. Under \emph{equity compensation}, a principal motivates the agents to work by giving them shares of the output. We describe the optimal equity allocation.  It is characterized by a \emph{neighborhood balance} condition: any two agents receiving equity have the same (weighted) total equity assigned to their neighbors. We also study the problem of selecting the team of agents who receive positive equity, and show this team must form a tight-knit subset of the complementarity network, with any pair being complementary to one another or jointly to another team member. Finally, we give conditions under which the amount of equity used for compensation is increasing in the strength of a team's complementarities and discuss several other applications.

\end{abstract}

\date{\today}
\maketitle

\section{Introduction}
\label{sec:introduction}

A popular method of motivating the members of a team to work toward a common goal is giving them shares of the output that their collaboration will achieve---\emph{equity compensation}. How should a principal design equity compensation for a team? How should this design depend on the structure of the team?  

We analyze this problem in a standard network game model of production with heterogeneous complementarities. A group of workers pursues a joint project, which succeeds with a probability depending on the team's performance and fails otherwise. Each worker chooses a hidden effort, represented as a nonnegative real number, at a convex cost. Each worker's effort makes a ``standalone'' contribution to the team performance, independently of anyone else's effort. But there are also production complementarities: the impact on team performance of a worker's effort is increasing in the efforts of  specific other, complementary, workers.  In our model, the pattern of these complementarities is represented by an arbitrary, exogenously given \emph{network}. We adapt a canonical model of the effort-provision game, in which  equilibrium efforts are uniquely 
determined by tractable formulas once individuals' incentives for effort are specified \citep*{BCZ-06}.

Our contribution is to solve for the principal's optimal contract, which takes the form of an allocation of equity stakes in the firm's returns. More precisely, there is a principal who allocates the output of a project in case of success, giving each worker a certain percentage. These payments, along with the technology of production described above, induce a network game, whose equilibrium outcomes determine effort levels and team performance.  The principal's problem is therefore to design a contract in this moral hazard setting with network spillovers. The effective network of incentive spillovers endogenous that governs equilibrium effort depends on the exogenously given complementarities but is also shaped by the principal's choice of contract. This makes solving the principal's problem much less straightforward than analyzing a network game with exogenously given spillovers.

Our results can be divided into three main categories. The first set concerns the intensive margin: among those agents who get equity, how much should they get? How do their positions in the network of complementarities determine their optimal equity pay? The second set concerns the extensive margin. It turns out that it need not be optimal to give all agents positive equity shares, and the optimal team induced to work may be a subset of the potential contributors. We study the problem of choosing this optimal team. Then, armed with a characterization of optimal equity allocations on both the intensive and extensive margins, the third set of results focuses on some implications of applied interest, examining how equity shares and outcomes vary with the network and the strength of complementarities.

\subsection*{Intensive margin.} We first study the intensive margin of the principal's contracting problem. \emph{Active} agents in a certain allocation are defined to be those who receive positive equity; these turn out to be precisely those agents who exert positive effort in equilibrium (since the incentive to work comes only from equity pay). Our first set of results analyzes the optimal allocation among these agents, showing that it satisfies a simple \emph{neighborhood balance} condition. Consider the weighted sum of equity shares allocated to the neighbors of an active agent, where the weights are the complementarities between an agent and the neighbor. Our first result, which underlies our analysis of the intensive margin, states that this weighted sum does not depend on the agent's identity---i.e., the weighted sum is the same for all active agents. 

A notable property implied by this is that agents' equilibrium efforts under it are proportional to their equity shares. This property is interesting in that it does not hold for a generic equity allocation, but we show that it must be satisfied for the principal's allocation to be optimal. The optimal allocation also features \emph{balanced neighborhood actions}: the total actions of an agent's neighbors, weighted by the strength of their complementarities, is the same for all active agents. We show the balanced neighborhood equity and action conditions are together equivalent to the principal's first-order conditions. Satisfying these turns out to balance the spillover effects of eliciting more effort
across active agents and to equalize the value of inducing any active agent to work more.

The endogenous equalization of incentives and activity across active neighborhoods is notable. In standard models of games with network complementarities (see below for a brief discussion of the literature), a theme is that more central agents and communities endogenously have higher incentives and higher neighborhood activity. In our setting, within the set of active agents, this inequality is endogenously muted: the principal has incentives to allocate equity so as to balance out both incentives and activity at the neighborhood level in the way we have described. {In our discussion of the literature below, we comment on how this changes the analysis relative to canonical studies of network games.}

The neighborhood balance result permits an explicit characterization of the share of equity each agent receives under the optimal allocation. The result can be reformulated as stating that an agent's share is proportional to a certain measure of that agent's centrality in the subnetwork of active agents. A vector $\mathbf{x}$ is called an \emph{equity centrality vector} for a network with weight matrix $\mathbf{W}$ if $\mathbf{W} \mathbf{x} = \mathbf{1}$, where $\mathbf{1}$ is the column vector of all ones. Equity centralities exist and are uniquely determined ($\mathbf{x}=\mathbf{W}^{-1}\mathbf{1}$) whenever $\mathbf{W}$ is invertible---a property that holds generically.  

\subsection*{Extensive margin.} The above results fully characterize optimal equity allocations and associated equilibria given an active set, i.e., a set of agents receiving positive equity shares at the optimal allocation. To characterize optimal contracts fully, however, we must also optimally choose the active set, which  need not contain all agents. This discrete optimization problem requires new insights beyond the intensive margin analysis.

We approach the problem by first reducing it to a quadratic program. This allows us to deduce two results implying considerable structure on the active set. Substantively, we find that highly connected subnetworks are optimal for the principal. Our first result on this is that the active set under any optimal allocation always has diameter at most two in the complementarity network, meaning that any two active agents either have complementarities with each other or both have complementarities with some shared active neighbor. Our interpretation of this result is that an optimal active set should be sufficiently ``tight-knit.'' We give a rough intuition for this result. When equity is given to members of a tight-knit group, the incentives given to one member also motivate effort by the others due to spillovers. On the other hand, if a given amount of equity is given to two subsets of agents that are not tightly linked in the complementarity network, then the equity given to one subset dilutes the incentives of the other without a strong counteracting beneficial effect of spillovers. Thus, a principal prefers to ``concentrate'' incentives and focus them on a single highly complementary group. 
This force is seen even more sharply in our second main result in this section, when we restrict attention to the standard benchmark of unweighted networks---in which all non-zero complementarities have the same strength. In this case, any maximum clique (a subnetwork with complementarities between all pairs of agents) is an optimal active set. 

Our overall interpretation of these results is that when a firm relies on a single joint outcome to provide incentives, teams with dense or tightly-knit complementarities outperform more dispersed teams. A further implication is that the principal often prefers to make a small team exert large efforts in order to make the most use of complementarities, rather than eliciting less effort from a larger group with more diffuse complementarities. Under convex effort costs, this can entail a considerable loss in agent welfare compared with other policies yielding somewhat less output.

\subsection*{Implications.}  Our final set of results explores some implications of our theoretical characterizations motivated by applied questions about the structure of teams and their compensation.

An important set of questions concerns how the optimal allocation and the probability of success depend on the network and the strength of complementarities. Our expression for equity centrality lets us calculate explicitly how the optimal shares vary as the network changes.  These comparative statics show that equity centrality can behave quite differently from measures such as Bonacich centrality, relevant in network games with exogenous incentives. In particular, monotonicity properties, whereby strengthening one's network links necessarily increases centrality, do \emph{not} hold for equity centrality. We illustrate this non-monotonicity and others in examples of three-agent networks. An implication is that investments that strengthen complementarities, though they may be beneficial for aggregate output, are not necessarily in agents' own interests---and this can occur even if the investments are not costly for the agents. 

Questions of network design are interesting more generally. Which links are most valuable to the principal? To make some progress toward understanding this issue, we ask how strengthening links between agents affects the team's probability of a successful outcome. The value to the team of strengthening a link takes a surprisingly simple form: it is proportional to the product of the equity shares allocated to the two agents. So if the principal can strengthen some complementarities, it is most valuable to focus on connections between agents who are already (equity) central. 

A final, practically important, question is how much equity a firm should devote to compensation. The trade-off is that allocating more equity to the team elicits more effort from them, but the principal gets a smaller share of the resulting pie. How a principal manages this trade-off depends on the environment, including the network and the strength of complementarities. Our last result gives conditions under which a profit-maximizing firm wants to distribute more shares to agents when complementarities are stronger. Intuitively, greater complementarities in production increase the returns to allocating shares to agents, since each share now drives more additional effort through spillovers. Establishing this for general complementarity networks turns out to be subtle. 

\subsection*{Literature}
We close with a brief discussion of relevant literature and the nature of our contribution. The literature on network games is extensive, going back to seminal papers including \citet*{goyal2003networks}, \citet{BCZ-06}, and \citet*{galeotti2010network}. The quadratic model of effort that we study has emerged as a focal point due to its tractability, connections to network centrality measures that are of independent interest, and amenability to empirical work. Much recent research has occurred even since the latest major surveys, such as \citet{jackson2015games}, \citet{bksurvey}, and \citet{zenou2016key}.\footnote{For some recent work on this type of model and closely related ones, see, e.g., \citet*{bramoulle2022loss}, \citet*{frick2022dispersed}, 
 \citet*{matouschek2022organizing}, and \citet*{cerreia2023dynamic}.} However, the study of how network complementarities interact with the design of incentive schemes---while a topic of obvious theoretical interest and practical relevance---is in its early stages. For example, \citet{belhaj2018} studies a problem where the principal must target a single agent and offer a contract. \citet{shi2022optimal} studies a model in which the network affects output via a distinct ``helping effort'' that agents can exert to change others' marginal costs of effort (rather than direct complementarities as in our model).   

We contribute to this literature both by posing  a simple optimal contracting  problem in a canonical network games model and by deriving a sharp description of incentives and behavior at the optimum quite different from any appearing in the works just mentioned. At a technical level, our problem has an interesting complication. Most network game analyses have a fixed network of spillovers, describing how an agent's optimal action (or some other analogous variable) depends on others' actions. To the extent that planner interventions of various kinds on nodes' incentives are studied in these network game models---as for example in \citet*{galeotti2020targeting}, \citet*{leister2022social},  or \citet{parise2023graphon}---the interventions typically affect a node attribute and do not change the spillover network. However, in our setting,   when a principal varies the equity stakes that different agents hold, the effective network of spillovers determining equilibrium behavior also changes. For example, when an agent gets a larger share of the group output, he cares more about the joint output with every collaborator, making him more strategically sensitive to those collaborators' efforts. The resulting endogeneity makes the principal's optimization problem substantially richer than it would be with a fixed spillover network. It is therefore interesting that the model nevertheless affords a simple characterization of optimal interventions in terms of the exogenous complementarity network.

Our work also ties into a large economics literature on the design of incentives, going back to Holmstrom's [\citeyear{holmstrom1982moral}] seminal contribution on incentives for teams when individual effort is not observable or not contractible. Within this literature we are closest to \cite{bernstein2012contracting}, which analyzes optimal incentives to induce all agents to exert effort in a binary-action game with network spillovers.\footnote{Related implementation problems are studied in \cite*{halac2020raising}, who consider heterogeneous agents without a network structure, and \cite{Lu2022}, who allow network monitoring rather than network spillovers.} 
We note two differences. First, the binary-action game in \cite{bernstein2012contracting} admits multiple equilibria for many parameters, and the focus of their analysis is full implementation of the maximal action profile. We instead consider a framework with a unique equilibrium and design incentives that maximize performance at that equilibrium. Second, network structure affects the optimal contract in \cite{bernstein2012contracting} primarily through asymmetry in links: on undirected networks, there is a lot of multiplicity in optimal contracts---with one such contract for every possible ranking of agents. In our model, by contrast, optimal contracts on a given undirected network depend intricately on the network structure and typically must motivate particular agents more than others.

\section{Model}
\label{sec:model}

We consider a model with one principal and $n$ agents, $N=\{1,2,\ldots,n\}$. These agents take real-valued actions $\actionagt{i} \geq 0$. Denote the joint action profile by $\actions=(\actionagt{1},\dots,\actionagt{n})$. To represent the complementarities among the agents, we define a weighted network with adjacency matrix $\network$, so $G_{ij} \geq 0$ is the weight of the link from $i$ to $j$. The neighborhood of agent $i$ is $\nghbd{i} = \left\{ j: G_{ij} > 0\right\}$. We call a network unweighted if $G_{ij} \in \{0,1\}$ for all $i$ and $j$.

Agents jointly work on a project which either succeeds or fails. The project outcome depends on agents' actions $\actions$, the network $\network$, and a parameter $\beta > 0$ measuring the strength of complementarities between agents. Let $S \in \{0,1\}$ be a binary variable corresponding to project success. We assume the probability of success is $P(\production)$, where $\production(\actions)$ is called the \emph{team performance} and $P: \mathbb{R}_{\geq 0} \rightarrow [0,1)$ is strictly increasing, concave, and twice differentiable. The team performance is the sum of a term that is linear in actions---corresponding to agents' standalone contributions---and a quadratic complementarity term:
$$\production(\actions) = \sum_{i\in N}\actionagt{i} + \frac{\beta}{2}\sum_{i,j \in N} G_{ij}a_i a_j. $$
A successful project produces an output, whose value we normalize to one, whereas a failed project produces a value equal to zero.

Throughout, we take the matrix $\network$ to be symmetric, or equivalently the network to be undirected. Because all payoffs will only depend on $\network$ through the team performance $\production(\actions)$, this assumption is without loss of generality (as we can replace $\mathbf{G}$ with $(\mathbf{G}+\mathbf{G}^T)/2$ without changing team performance). We also assume $\network$ is not identically zero.

The principal observes the project outcome but does not observe agents' actions. (When we use pronouns, we use ``she'' for the principal and ``he'' for an agent.) To incentivize effort, the principal offers a contract, which specifies a non-negative transfer $t_i(S)$ to each agent that can depend on the project outcome.  Agents maximize the expectation of the following payoff, which is quasi-linear in monetary transfers and has a quadratic cost of effort:  $$ u_i = t_i(S) - \frac{a_i^2}{2}.$$ 
The environment (including the network and contract) are common knowledge among agents, and the network is known to the principal when she is choosing the contract.

\subsection{Objectives}\label{s:objectives}

We will see in \Cref{sec:uniqueeq} that, given any contingent payments, there is a unique Nash equilibrium, which we call $\eqlbactions$. The principal optimizes over contracts, expecting this equilibrium to be played. In this section, we define two objectives for the principal; our results will apply to both objectives, except where we explicitly state otherwise.

Under the \emph{residual profit} objective, the principal maximizes the expected value of project success (normalized to $1$) minus payments to agents $$ P(Y(\eqlbactions))-\mathbb{E} \left[ \sum_{i \in N} t_i(S) \right].$$

It turns out to be optimal to give all agents a transfer of zero when the project fails (as we will argue formally in \Cref{sec:uniqueeq}), so without loss of optimality we can consider contracts as transfers $\sharesvector = (\sharesagt{1},\dots,\sharesagt{n})$ to each agent when the project succeeds. We will interpret these payoffs as equity shares in the project. Rewriting the expected residual profit of the principal, the residual profit maximization problem can be written as
\begin{equation} \text{choose $\sharesvector$ to maximize } 
 \residualprofits(\sharesvector) = \left(1-\sum_{i=1}^{n}\sharesagt{i}\right)P(\production(\mathbf{a}^*)). \tag{RP} \end{equation}

Under the \emph{success probability} objective, the principal maximizes the probability of success $P(Y(\eqlbactions))$ subject to the constraint that the total transfers are no larger than the output from the project. This constraint rules out positive transfers when the project fails. Thus, the success probability  maximization problem is \begin{equation} \text{choose $\sharesvector$ to maximize } 
 P(\production(\mathbf{a}^*)) \text{ subject to } \sum_i \sigma_i \leq 1. \tag{SP} \end{equation} The $\sigma_i$ can again be interpreted as equity shares.

An alternative model, which turns out to be very similar, is that the project yields a deterministic monetary output $P(Y)$, where $P: \mathbb{R}_{\geq 0} \rightarrow \mathbb{R}_{\geq 0}$  is a strictly increasing, concave, and twice differentiable function. Our analysis applies essentially unchanged to the study of this production function with either objective, provided that the principal uses \emph{linear} contracts, i.e., giving each agent a transfer equal to a fixed equity share $\sigma_i$ of output.\footnote{If $P$ is not bounded, this deterministic model requires an additional assumption that $\beta$ is small enough so that the principal's feasible payoffs are bounded.} In the rest of the paper, we will work with the binary-outcome model.

\subsection{Discussion of modeling assumptions}

There are several dimensions of our modeling assumptions that are worth commenting on. First, we assume that the project outcome---the random variable $S$---is the only observable consequence of any agent's effort. In other words, agents cannot be paid directly for their efforts $a_i$. The motivation behind this modeling assumption (discussed in \citet{holmstrom1982moral} and the ensuing literature) is that many aspects of agents' productive efforts are not observable, or not possible to make binding legal commitments over. 

One could enrich the model with more general signals about agents' efforts that a contract could condition on.
Our basic model highlights relationships between networks and equity pay that we expect would be relevant in such extensions, as long as equity pay was an important part of motivating (the marginal unit of) effort.

Second, we consider a binary-outcome model where the contracts available to the principal are essentially equity schemes, in which each agent's compensation is a share of the output of the firm (see \cite{tirole2012overcoming}, for example, for a similar modeling technique).
Our main motivation for being interested in equity payments is that this is an extremely popular form of incentive in certain types of organizations, along with closely related instruments such as options (see, e.g., \cite{levin2005profit}). Various models have been used to analyze the reasons for using equity pay as opposed to other contracts when they are available \citep{holmstrom1987aggregation, dai2022robust}.

\section{The intensive margin: Optimal equity shares \\and actions for active agents}
\label{sec:results}

This section characterizes the optimal allocation of equity among those who receive positive shares, as well as  the induced equilibrium. The first subsection describes the unique equilibrium of the network game given a fixed contract. The second subsection states our first main result, which describes the equity shares under the optimal contract and the corresponding equilibrium actions, and discusses various economic consequences.

\subsection{Equilibrium of the network game}\label{sec:uniqueeq}

We now show equilibrium is unique given any contract $(t_i)_{i\in N}$ and provide a characterization of equilibrium actions. Because agents' incentives depend only on the difference $t_i(1)-t_i(0)$ between transfers conditional on success and failure, we can shift payments and assume $t_i(0)=0$ without loss of generality for proving uniqueness. Similarly, this shift can only improve the principal's payoff, so it is without loss of optimality in the principal's problem. Thus, from now on, we will let contracts be described by equity shares $\sharesvector$. Fixing such a vector (which need not be optimal), agents' payoffs are
$$U_{i}(\actions,\network,\boldsymbol{\sigma}) = P\left(\production\right)\sharesagt{i}-\frac{\actionagt{i}^{2}}{2} .$$
Since agents receive shares of the team's output, their marginal returns to effort depend on others' actions. The first-order conditions for agents' best responses are
$$\actionagt{i} = P'(Y)\sharesagt{i}\left(\beta \sum_{j}G_{ij}\actionagt{j} + 1\right).$$
The following result states that these first-order conditions characterize the  unique Nash equilibrium.
\begin{proposition}\label{p:uniqueeq}
Fixing $\sharesvector$, there exists a unique Nash equilibrium. The equilibrium actions $\eqlbactions$ and  team performance $Y^*$ solve the equations
\begin{equation} [\mathbf{I}-P'(Y^*)\beta \sharesmatrix \network]\eqlbactions = P'(Y^*)\sharesvector \text{ and }Y^* = \production(\eqlbactions),\label{eq:Nash_conditions} \end{equation} where $\sharesmatrix = \operatorname{diag}(\sharesvector)$ is the diagonal matrix with entries $\Sigma_{ii} = \sharesagt{i}$. 
\end{proposition}

Note that the result entails a positive equilibrium action for those agents with $\sigma_i>0$, and an action of zero otherwise. 

For intuition, suppose that $P(Y)= \alpha Y$ for all relevant action vectors, where $\alpha>0$ is a constant. Then equilibrium actions are given by
\begin{equation} \eqlbactions = [\mathbf{I}-\alpha \beta \sharesmatrix \network]^{-1} \alpha \sharesvector.  \label{eq:Nash_conditions_linear} \end{equation}
 When all agents receive equal shares, equilibrium actions are proportional to Bonacich centralities (as in \citet{BCZ-06}). For arbitrary shares, the actions are a modified version of Bonacich centrality with respect to a network $\sharesmatrix \network$.

The network $\sharesmatrix \network$ reflects spillovers; its $(i,j)$ entry is the slope of $i$'s best-response in $j$'s action. One can see from the form of this matrix that it is endogenous to the principal's choice of  equity compensation, as discussed in the introduction: when an agent gets a larger share of the group output, the agent cares more about the joint output with every collaborator, making the agent more strategically sensitive to those collaborators' efforts.

\subsection{Balance conditions in the optimal contract}

Our first main result characterizes the optimal allocation and equilibrium actions among the set of agents receiving positive shares.

We have noted that an agent exerts positive effort under a given contract if and only if he receives positive equity. We will thus call an agent \emph{active} under a given equity allocation $\bm{\sigma}$ if he receives a positive equity share $\sigma_i>0$ and \emph{inactive} otherwise.

\begin{theorem} Suppose $\sharesvector^*$ is an optimal allocation and $\eqlbactions$ and $Y^*$ are the induced equilibrium actions and team performance, respectively. The following properties are satisfied:
    \label{t:optsharescharacterization}
    \begin{enumerate}
        \item \textbf{Balanced neighborhood equity}: There is a constant $c>0$ such that for all active agents $i$, we have $(\network \sharesvector^*)_i = c$.
    \item \textbf{Actions are proportional to shares}: $\eqlbactions  = \mu \sharesvector^*,$ where $$ \mu  =  \frac{P'(Y^*)}{1-P'(Y^*)\beta c}.$$
    \item \textbf{Balanced neighborhood actions}: For all active agents, $(\network \eqlbactions)_i = \mu c$.
    \end{enumerate}
\end{theorem}

Omitted proofs are in \autoref{a:appendix_proofs}. The conditions in the theorem state, at a high level, that it is optimal for the principal to equalize the complementarities motivating various active agents to work. We will give an intuition for why this is a necessary feature of an optimal contract below, but first we spell out the content of the result and some simple relationships among its parts.

The property of balanced neighborhood equity says that for each active agent $i$, the sum $\sum_j G_{ij} \sigma_j$ of shares given to neighbors of $i$, weighted by the strength of $i$'s connections to those neighbors in $\network$, is equal to the same number (i.e., is not dependent on $i$). %

Part (b) says that under the optimal allocation, each active agent's equilibrium effort is a constant multiple of the agent's equity share. This turns out to follow from part (a), as we now explain. For a given optimal share vector $\optsharesvector$, let us
write $\optsharesmatrix=\diag(\optsharesvector)$; then (a) is equivalent to saying that $\optsharesvector$ is a right-hand eigenvector of $\optsharesmatrix \network$ with eigenvalue $c$. Now consider the equation we get when we solve (\ref{eq:Nash_conditions}) for $\mathbf{a}^*$: $$\eqlbactions = \underbrace{P'(Y^*)\left[\mathbf{I}-P'(Y^*)\beta \optsharesmatrix \network \right]^{-1}}_{\mathbf{M}} \optsharesvector.$$ Our observation about $\optsharesvector$ implies that it is also an eigenvector of the matrix $\mathbf{M}$ with eigenvalue $\mu$; that establishes (b).\footnote{An alternative argument is to expand the expression for $\eqlbactions$ using the Neumann series to write $\eqlbactions=P'(Y^*)\sum_{k=0}^{\infty}P'(Y^*)^k\beta^{k}\left(\optsharesmatrix \network\right)^{k} \optsharesvector$; then repeatedly use balanced neighborhood equity to rewrite this as $P'(Y^*)\sum_{k=0}^{\infty}P'(Y^*)^k\beta^{k}c^{k}\optsharesvector$.}

The property of balanced neighborhood actions states that for each active agent $i$, the sum of actions of neighbors of $i$, weighted by the strength of $i$'s connections to those neighbors in $\network$, is equal to the same number, $\mu c$. This follows immediately from (a) and (b).

The system of equations in part (a) of \autoref{t:optsharescharacterization} can be solved explicitly for the optimal shares $\optsharesvector$ as long as the relevant adjacency matrix is invertible, which holds for generic weighted networks. Motivated by this, we define equity centrality:
\begin{definition}
Given a weighted network  with non-singular adjacency matrix $\mathbf{W}$, the \textbf{equity centrality} of agent $i$ is $\left(\mathbf{W}^{-1}\boldsymbol{1}\right)_i$.
\end{definition}
\autoref{t:optsharescharacterization} then entails that under an optimal allocation, for each active agent $i$, the  equity share $\sigma_i^*$  is proportional to $i$'s equity centrality in the subnetwork $\subnetwork$ of active agents for that allocation; the same is true for actions, with a different constant of proportionality.

Equity centrality behaves quite differently from standard measures such as Bonacich centrality. In particular, the inverse $\mathbf{W}^{-1}$ changes non-monotonically as $\mathbf{W}$ changes. We will see in \Cref{sec:compstatics} that  this can induce non-monotonicities in the optimal allocation and the resulting actions and utilities.

An implication of our characterization of the optimal equity is that the ratio of shares allocated to two active agents is independent of the complementarity parameter $\beta$, and depends only on the network $\network$ and the set of active agents. (Our results in \Cref{sec:active} will imply that the optimal active sets are also independent of $\beta$.) Since the induced actions are proportional to shares, the ratio between the equilibrium actions of any two active agents is independent of $\beta$ as well. This property is surprising, because in standard network games analyses \citep{BCZ-06}, equilibrium actions are highly sensitive to $\beta$. This dependence is endogenously exactly canceled out by the planner's optimal equity allocation.

\begin{remark}
    The result of \autoref{t:optsharescharacterization} holds under either the residual profit or success probability objective, as will all our characterizations of optimal contracts. In fact, the proof of \autoref{t:optsharescharacterization} establishes a stronger statement: if $\optsharesvector$ solves the problem of maximizing $Y(\mathbf{a}^*)$ subject to the constraint  $\sum_i \sigma_i=s$, where $s$ is any positive number, then $\optsharesvector$ and $\mathbf{a}^*$ must satisfy the conditions given in the theorem.
\end{remark}

\subsubsection{Intuition for \autoref{t:optsharescharacterization}} We can provide some intuition for the balanced neighborhood equity and action conditions by informally arguing they are \emph{sufficient} for a certain principal first-order condition that must hold at the optimal contract. First, we can redescribe the principal's problem as choosing \emph{actions} from among those that can be implemented at some optimal contract. 
Suppose we perturb agent $i$'s action exogenously by $\epsilon$ and follow the consequences through the system of best responses. Each agent $j$'s best response is
$$P'(Y^*)\sigma_j \left(1+\beta \sum_{j'\in N} G_{jj'} a_{j'}\right),$$
 so the direct impact on $j$'s action, given by the $j'=i$ term, is to increase it by $\beta P'(Y^*)\sigma_j  G_{ji} \epsilon$, which (by symmetry of $\network$), is equal to $\beta P'(Y^*)\sigma_j  G_{ij} \epsilon$. The balanced neighborhood equity condition implies that the sum of these direct impacts does not depend on $i$. That is, the direct impact on the aggregate effort level does not depend on which agent's action we perturbed. Iterating this argument, the indirect impact of increasing $i$'s actions on the total of actions does not depend on $i$'s identity either. So the balanced neighborhood equity condition implies that increasing any agent's action marginally has the same effect on the total of all actions.

Two gaps remain between this indifference and the principal's first-order condition, which requires that redistributing shares among active agents locally does not affect output. First, some agents might increase their actions more than others when given $\epsilon$ additional equity. (So even if a principal is indifferent between any same-sized perturbation to actions, she may be able to achieve some of these more cheaply than others.) This problem does not arise precisely in case actions are proportional to shares, which is implied by the two conditions in the theorem. Second, output is not only the sum of shares; the quadratic term in output could change the principal's first-order condition. But it turns out the balanced neighborhood action condition also implies the first-order conditions for output is actually the same as for the sum of efforts.

We have argued that the two conditions in \autoref{t:optsharescharacterization} imply the principal's first-order condition is satisfied. It is not obvious that these two conditions---balanced neighborhood equity and balanced neighborhood actions---can be satisfied simultaneously; indeed, this depends on some of the specific structure of our model. The full proof establishes that these conditions can be jointly satisfied and that they are also \emph{necessary} conditions for an allocation to be optimal.

\section{The extensive margin: Active and inactive agents}
\label{sec:active}

We next ask which agents are active and which agents are inactive under optimal contracts for a given complementarity network. Recall an agent is defined to be active under a given allocation if he receives positive equity.

The main results in this section show that the active sets under optimal allocations are highly connected subnetworks. We first show that any optimal active set has diameter at most two in the complementarity network $\network$. We then show that in the special case of unweighted networks, there is an optimal allocation with a clique as the active set. We interpret these results as saying that, when incentives to exert effort are based only on global outcomes (and not local measures of performance), smaller and more highly connected teams outperform larger and more dispersed teams.

It is not immediate from \autoref{t:optsharescharacterization} which agents are active; indeed, there can be several candidate active sets compatible with the condition of the theorem. (\autoref{t:optsharescharacterization} implies that the candidate active sets are the subnetworks $\subnetwork \subseteq \network$ such that the row sums of $\subnetwork^{-1}$ are positive.)

The key to our analysis is the fact (formalized as \autoref{l:activesetopt} in \autoref{a:appendix_proofs}) that we can find the active sets under optimal allocations by solving the following optimization problem\footnote{This is a quadratically constrained optimization of a linear objective, since the  constraint can be rewritten as $\sigma_i(\mathbf{G}\boldsymbol{\sigma})_i=0$.} among non-negative equity shares summing to a fixed value:
\begin{equation}
  \begin{aligned}
    & \max_{\sharesvector} && c \\
    & \text{subject to} && (\network \boldsymbol{\sigma})_i = c \text{ whenever }\sigma_i > 0. \\
  \end{aligned} \label{eq:QP}
\end{equation}

The reason for this reduction is that, when the balanced equity condition holds, we can rewrite output as the following function of the total shares allocated to agents and $c$, the total weighted equity in each active agent's neighborhood: \begin{equation}
    \production(\eqlbactions) = \left(\sum_{i=1}^{n}\sharesagt{i}\right)\left(\frac{P'(Y^{*})}{1-\beta P'(Y^{*})c} + \frac{\beta P'(Y^{*})^{2} c}{2(1-\beta P'(Y^{*}) c)^{2}}\right).
    \label{eq:eqlbtotalproduction}
\end{equation}
The right-hand side is an increasing function of $c$. So the principal can choose an optimal active set by maximizing the constant $c$ in the balanced equity condition. 

One important implication of this is the following invariance of the principal's optimization to the complementarity parameter:
\begin{proposition}\label{p:invariance}If $\bm{\sigma}_0^*$ is a solution to the principal's problem under $\beta_0>0$ and $\beta_1$ is another complementarity parameter, then there is a constant $k>0$ so that $k\bm{\sigma}_0^*$ solves the principal's problem under $\beta_1$.  \end{proposition}

This says essentially that the active set and the ratios in which equity is optimally allocated are both independent of $\beta$. This follows from observing that the optimization problem to solve for the active set does not depend on $\beta$. For the success probability objective, the constant $k$ is equal to $1$. Under the residual profits objective, the principal may adjust the total share of output distributed to agents as $\beta$ changes (see \Cref{s:beta_comp_stat}).

Before turning to other general implications of the lemma, we study the example of a three-agent network. This example, which is the smallest interesting case of our model, shows that the active set can depend in non-trivial ways on network structure. We describe the optimal allocation here and provide details in \autoref{a:three_agent}.

\begin{example}
Consider a weighted network with three agents without  self-links (see \autoref{l:threeagentcomplete}). Since it is optimal to have both agents active in networks with two agents, three-agent networks are the smallest non-trivial example of our model.

Without loss of generality, we can assume $G_{12} \geq G_{13} \geq G_{23}$ and choose the normalization $G_{12}=1$, so that the adjacency matrix is 
$$\network = 
\begin{bmatrix}0 & 1 & G_{13} \\
1 & 0 & G_{23} \\
G_{13} & G_{23} & 0
\end{bmatrix}.$$

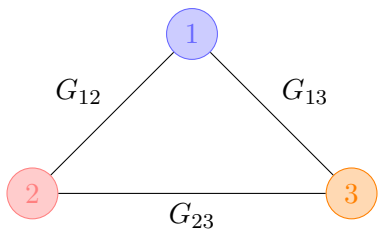
\begin{figure}
    \centering
        \begin{tikzpicture} [node distance={30mm}, agt1/.style = {draw,circle,color=blue!60,fill=blue!20}, agt2/.style = {draw,circle,color=red!50,fill=red!20}, agt3/.style = {draw,circle,color=orange!100,fill=orange!30}]
\node[agt1] (1) {$1$};
\node[agt2] (2) [below left of=1] {$2$};
\node[agt3] (3) [below right of=1] {$3$};
\draw (2) edge["$G_{12}$"] (1);
\draw (3) edge["$G_{23}$"] (2);
\draw (1) edge["$G_{13}$"] (3);
\end{tikzpicture}
    \caption{Three agent weighted graph with weights $G_{12}, G_{13},$ and $G_{23}$.}
    \label{l:threeagentcomplete}
\end{figure}

The optimal active set consists of either agents $1$ and $2$ or all three agents.  If all agents are active, then \autoref{t:optsharescharacterization} implies that the optimal shares must solve
$$\optsharesagt{1} = \frac{1+G_{13} - G_{23}}{2G_{13}} c, \quad \optsharesagt{2} = \frac{1+G_{23} - G_{13}}{2G_{23}} c, \quad \optsharesagt{3} = \frac{G_{23} + G_{13} - 1}{2G_{23}G_{13}} c,$$
for some constant $c$. Since we must have $\sigma^*_3 > 0$ for all agents to be active, a necessary condition for $\{1,2,3\}$ to be an optimal active set  is $G_{13} + G_{23} > 1$.

This necessary condition also turns out to be sufficient; we now sketch the argument for this. A calculation shows that distributing a total amount $s$ of equity consistent with these ratios gives a value of $$c = \frac{2G_{13}G_{23}s}{2(G_{23}+G_{13}) - 1 -(G_{13}-G_{23})^{2}}.$$
On the other hand, if a total amount $s$ of equity is distributed among two active agents, then $c=s/2$. A bit more algebra confirms that $$\frac{2G_{13}G_{23}s}{2(G_{23}+G_{13}) - 1 -(G_{13}-G_{23})^{2}} > \frac{s}{2}$$
whenever $G_{13} + G_{23} > 1$.

So the active set includes all three agents if and only if $G_{13} + G_{23} > 1$. In this example, the principal maximizes $(\mathbf{G}\sharesvector)_i$ subject to the balanced equity condition by choosing an active set that maximizes the minimum weighted degree of an agent in the induced subnetwork. If the least connected agent's complementarities are too weak, the principal prefers to exclude that agent from the team and concentrate on incentivizing the two agents with stronger complementarities.\qedhere
\end{example}

Moving back to general networks, we now state two results showing that the principal prefers a highly connected active set. The first result on this holds for any network satisfying our maintained assumptions.  It states that the network distance between any pair of agents in the active set is small. Recall that the diameter of a network is the longest distance\footnote{Shortest path consisting of links with positive weights.} between any two agents in the network.
\begin{proposition}\label{p:diameter}
The diameter of the active set under any optimal allocation is at most $2$.
\end{proposition}

The idea is that if two agents $i$ and $j$ are at distance larger than two from each other, their neighborhoods are disjoint. An optimal allocation must divide shares between the disjoint sets $\{i\}$, $\{j\}$, and $N(i)$, and $N(j)$ (as well as any other potential agents in the active set). It must also satisfy the balanced neighborhood equity condition, which implies that the shares allocated $N(i)$ and $N(j)$ cannot be very large. Indeed, the proof shows that any such allocation is dominated by allocating shares to only two agents: splitting shares evenly between two agents connected by a link with the largest weight in the network gives a higher value of the constant $c$.

In the special case of unweighted networks, the message that highly connected active sets are optimal can be sharpened: It is an optimal solution for the principal to choose any clique\footnote{A subnetwork with links between all pairs of agents.} of maximum size and divide shares equally among the agents in this clique.
\begin{theorem}
\label{t:unweighted_active_set}
If $\network$ is an unweighted network, then any maximum clique is the active set at an optimal allocation.
\end{theorem}

When connections are unweighted, choosing a subset that is as densely connected as possible leads to at least as high a payoff as choosing a larger but more sparsely connected subset, even if all agents in the larger subset have higher degree.

The proof applies \autoref{l:activesetopt}: given an optimal allocation with an arbitrary active set, we find a clique within that active set for which the constant $c$ is as large. We produce such a clique exists by sequentially constructing a series of agents who are all connected to each other. To do so, the balanced equity condition for the optimal allocation has to be used carefully at each step of the construction.

In general, the active set need not be unique in unweighted networks, so there can be other optimal allocations giving the same payoff to the principal as a clique of maximum size. For example, in a star network, any set including the central node and at least one peripheral node is the active set at some optimal allocation.\footnote{Among allocations giving $s$ shares to agents, a given allocation is optimal if and only if it gives $s/2$ shares to the central agent and $s/2$ shares to the peripheral agents.} The next example shows that there can also be optimal allocations that differ more substantially from maximum cliques.

\begin{example} \label{ex:multiple_active_sets}
Consider an even number of agents $n$ arranged in a circle. Let $\mathbf{G}$ be the undirected network in which each agent is connected to all other agents except the diametrically opposite agent in the circle:  for all distinct $i$ and $j$ we let $G_{ij} = 1$ if $|i-j| \neq n/2$ and $G_{ij}=0$ if $|i-j| = n/2$. The network structure is shown in \autoref{l:circle}.

Suppose that we want to allocate $s$ shares to agents optimally. The maximum cliques have size $n/2$, and dividing the $s$ shares evenly within any maximum clique gives $(\subnetwork\sharesvector)_i = \frac{n-2}{n} \cdot s$. All agents in the full network have degree $n-2$, so dividing $s$ shares evenly among all agents also gives $(\subnetwork\sharesvector)_i = \frac{n-2}{n} \cdot s$. It follows that the set of all agents, as well as all maximum cliques, are possible active sets, depending on which optimal allocation is chosen.
\end{example}

\begin{figure}
    \centering
      
\begin{tikzpicture}
\foreach \x in {0,...,9} {
    \pgfmathtruncatemacro\num{\x+1}
    \node[circle, draw=black, fill=white, thick] (n\x) at ({360/10*\x}:2.5cm) {\num};
}

\foreach \x in {0,...,9}
    \foreach \y in {1,...,4} {
        \pgfmathtruncatemacro\result{{mod(\x+\y,10)}}
        \draw[shorten >=1pt, shorten <=1pt] (n\result) -- (n\x);
    }
\end{tikzpicture}

    \caption{Ten agent unweighted graph with each agent connected to all other agents except the diametrically opposite one.}
    \label{l:circle}
\end{figure}
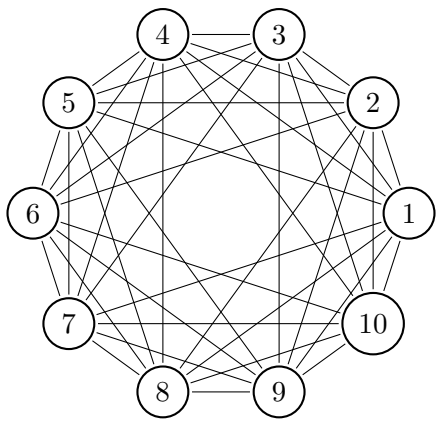

This example shows that the principal can be indifferent between very different active sets---a point which has implications for the welfare of agents, as we discuss further in \Cref{s:welfare}.

\section{Implications}
\label{sec:compstatics}

In this section, we explore some implications of our analysis that illuminate how the optimal contract depends on the environment. We focus on the effects of changes in the network $\network$ and the parameter $\beta$ describing the strength of complementarities. \Cref{sec:vary_network} examines how the equity allocation and the resulting team performance depend on the network of complementarities. The results provide some insights about which networks might be preferred by the principal and by agents. \Cref{s:beta_comp_stat} then asks how the share of equity retained by the principal (in the residual profit maximization problem) depends on the strength of complementarities. 

\subsection{Varying the network} \label{sec:vary_network}

Our first result describes how optimal allocations vary as the network changes. We write $\frac{\partial }{\partial G_{jk}}$ for the derivative in the weight $G_{jk}=G_{kj}$ of the link between $j$ and $k$. Recall that given an allocation, we write $\subnetwork$ for the adjacency matrix restricted to active agents.

\begin{proposition}
\label{p:compstaticmatrixcalc}
Suppose that under $\network$ there is a unique optimal equity allocation\footnote{We expect this hypothesis to be satisfied for generic networks.} $\optsharesvector$, with agents $i$, $j$, and $k$ all active. The derivative of agent $i$'s optimal share as we vary the weight of the link between $j$ and $k$ is
$$\frac{\partial \sharesagt{i}^*}{\partial G_{jk}} = - (\subnetwork^{-1})_{ik}\sigma_j^* - (\subnetwork^{-1})_{ij}\sigma_k^*+ \frac{\partial c}{\partial G_{jk}} \frac{\sigma_i^*}{c} .$$
\end{proposition}

The value $c$ is the balanced equity in each neighborhood from \Cref{t:optsharescharacterization}(a). The proof is based on the matrix calculus expression \begin{equation}\label{e:matrixinversederiv}\frac{\partial \network(t)^{-1}}{\partial t} = - \network(t)^{-1} \frac{\partial \network(t)}{\partial t} \network(t)^{-1}\end{equation} for the derivative of the inverse of a matrix.

The result provides a fairly explicit expression for the impact of changing a link on equity allocations. However, calculating the change $\frac{\partial c}{\partial G_{jk}}$ in $c$ may be difficult under the residual profits objective, where the total amount of equity allocated can change as we strengthen a link. We can be more explicit under the success probability objective, since the total amount of equity allocated sums to $1$. In that case, we have an explicit expression $$c = \frac{1}{\textbf{1}^T\subnetwork^{-1}\textbf{1}}$$
for the total equity in each neighborhood. Differentiating this expression gives a version of \autoref{p:compstaticmatrixcalc} without an unknown value $c$. 

Under either objective, the ratio $\sigma_i^*/\sigma_{i'}^*$ between two agents' shares is independent of $c$, so (\ref{e:matrixinversederiv}) lets us calculate the change in this ratio (under either objective) as the link between $j$ and $k$ is strengthened.

The matrix inverse $\network^{-1}$ need not vary monotonically as $\network$ changes. This implies that equity centrality need not satisfy monotonicity properties that hold for standard centrality measures such as Bonacich centrality. As an illustration, we return to our example of three-agent networks from  \Cref{sec:active}. Details are again deferred to \autoref{a:three_agent}. 

\setcounter{example}{0}

\continueexample{

Recall that we normalized $G_{12}=1$. We will vary the weight $G_{23}$ over the interval $(1-G_{13},G_{13})$.  We will further suppose $G_{13}>\frac12$; under these conditions, there is a unique optimal allocation and all three agents are active under this allocation.

We begin by studying the effect on optimal shares $\sigma_i^*$---or, in other words, equity centralities. Under the success probability objective, we show these centralities can be non-monotonic in an agent's own links. There exists a threshold $g^* \in (1-G_{13},G_{13})$ such that for $G_{23} \in (1-G_{13},g^*)$, increasing the weight $G_{23}$ decreases the share $\sigma_2^*$ allocated to agent $2$. So strengthening one of an agent's links can decrease his share of output under the optimal allocation. Intuitively, as $G_{23}$ is strengthened, the principal would like to increase agent $3$'s shares, and initially is willing to do so at the expense of agent $2$. (When $G_{23} \in (g^*, G_{23})$, meanwhile, increasing this weight decreases the share $\sigma_1^*$ allocated to agent $1$.)

Under the residual profit objective, comparative statics are more challenging because of the additional choice of how much equity to allocate (corresponding to the last term of the formula in \Cref{p:compstaticmatrixcalc}).  We turn to a numerical example to illustrate that non-monotonicities like those discussed above can nevertheless continue to be present. \autoref{fig:optshares_payoffs} shows the optimal equity shares and the corresponding equilibrium payoffs as we vary the link weight $G_{23}$, under parameter values specified in the caption. 
\autoref{fig:optshares} depicts the optimal equity allocation of each agent as a function of $G_{23}$. The equity allocation is non-monotonic in own links: increasing $G_{23}$ initially decreases agent $2$'s equity, mirroring our analytical result described above.

The numerical example also illustrates a corresponding non-monotonicity in payoffs: strengthening one of an agent's links can decrease his equilibrium payoff under the optimal contract. \autoref{fig:optpayoff} depicts the equilibrium payoffs under the optimal equity allocation as a function of $G_{23}$. Strengthening the link between agents $2$ and $3$ can \emph{decrease} the resulting payoffs for agents $1$ and $2$. This contrasts with the standard network games intuition: under a fixed equity allocation, all agents' payoffs are monotone in the network. In the present setting, however, agent $2$ can benefit from weakening one of his links. This suggests a tension between the network formation incentives of the principal and the agents. Agents may not be willing to form links that would benefit the principal or the team as a whole, even if link formation is not costly.

\begin{figure}
    \centering
    \begin{subfigure}{0.7\textwidth}
    \centering
    \includegraphics[width=\textwidth]{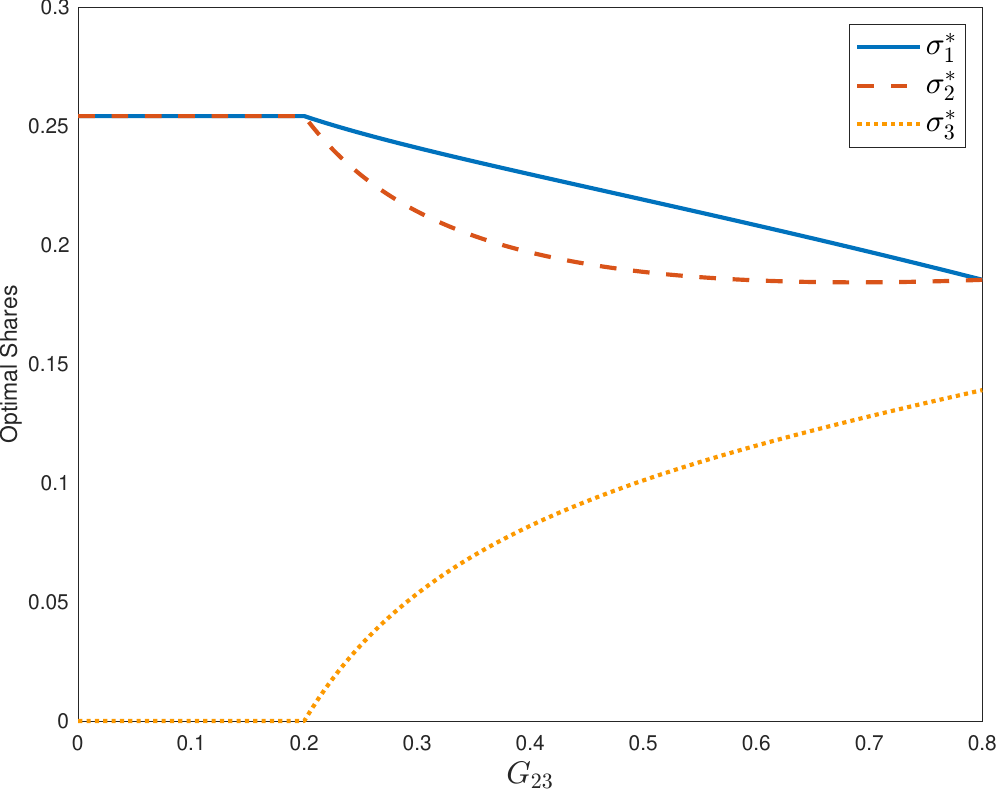}
    \caption{Optimal shares}
    \label{fig:optshares}
    \end{subfigure}
    \begin{subfigure}{0.7\textwidth}
        \centering
     \includegraphics[width=\textwidth]{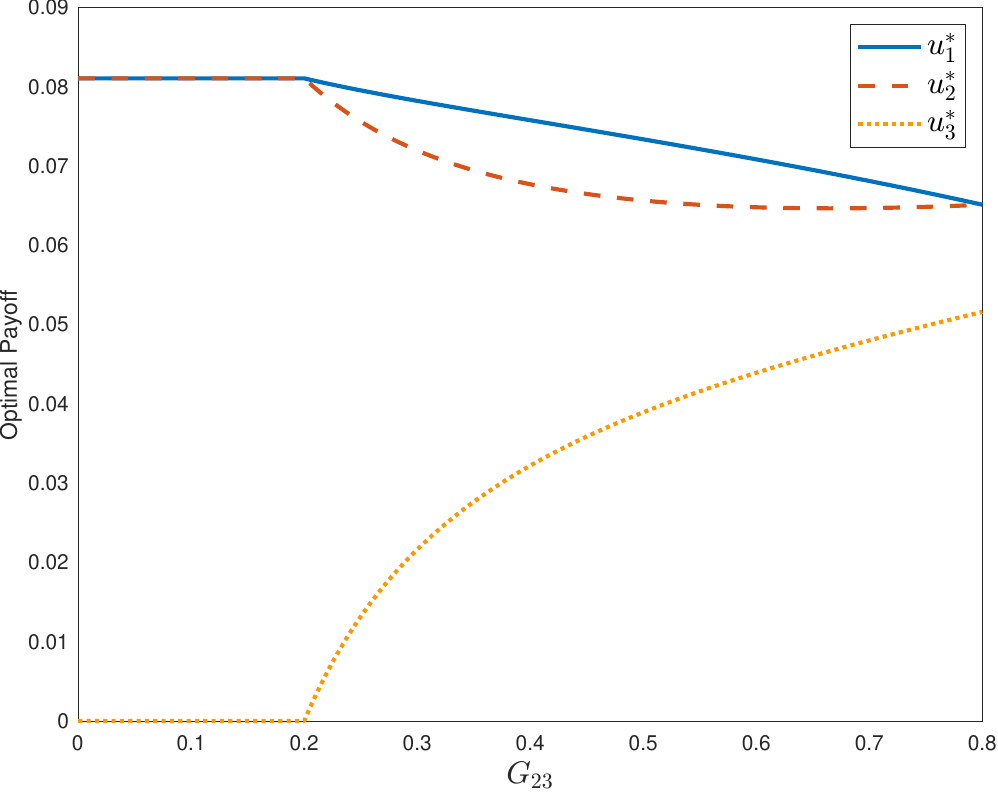}
     \caption{Payoffs under optimal shares}
     \label{fig:optpayoff}
    \end{subfigure}
    \caption{The optimal share allocation and resulting equilibrium payoffs as a function of the weight $G_{23}$. We work with the residual profits objective. Here $G_{13}=0.8$ and $\beta=0.1$, while $P(Y) = \min\{0.9Y,1\}$ (the kink is not relevant for the principal's problem). In both diagrams, the curve corresponding to agent $1$ is the topmost (solid blue) one; the curve corresponding to agent $2$ is the second from the top (dashed red); and the curve corresponding to agent $3$ is the lowest (dotted orange) one.}
    \label{fig:optshares_payoffs}
\end{figure}
}

We next look at how team performance under an optimal allocation varies as the network changes. Recall that $Y^{*}$ denotes the equilibrium team performance under an optimal allocation. Then $\frac{\partial Y^{*}}{\partial G_{ij}}$ is the change in this team performance as the weight on the link between agent $i$ and $j$ increases.

\begin{proposition}
    \label{p:compstaticoveralloutput}
    Suppose $\optsharesvector$ is an optimal allocation. Then the change in equilibrium team performance as $G_{ij}$ varies can be expressed as
    $$\frac{\partial Y^{*}}{\partial G_{ij}} = \sigma^{*}_{i}\sigma^{*}_{j}  h,$$
    where $h$ does not depend on the identities of $i$ or $j$.
    
    \end{proposition}

The proposition says that the increase in output from strengthening a link is precisely proportional to the product of the equity shares given to the two agents connected by that link. The proof gives an explicit formula for the quantity $h$, which depends on the model parameters and the allocation.

The proposition has implications for a designer who can make small changes in the network of complementarities. If the principal could marginally strengthen some links, she would want to focus on links between pairs of agents with high equity centralities. This is consistent with the intuition from \Cref{sec:active} that highly connected groups of agents are especially productive under equity compensation.

\subsection{Varying complementarities}\label{s:beta_comp_stat}

We now turn to how outcomes change as the complementarity parameter $\beta$ increases. Recall from \Cref{p:invariance}
that ratios of optimal shares do not depend on the value of $\beta$. But under the residual profits objective, we can ask how the total fraction of shares allocated to agents depends on $\beta$.

We study the comparative static in the special case when $P(\cdot)$ is linear in the range of feasible team performances. We assume for simplicity that the optimal allocation is unique, but could easily relax this assumption. The principal faces a trade-off between keeping a larger share of the profits and using a larger share to encourage workers to exert more effort. The following result states that when complementarities in production are larger,  it is optimal to keep a smaller share of a larger pie.

\begin{proposition}
    \label{p:compstatictotalequity}
    Suppose that $P(Y) = \alpha Y$ on an interval $[0,\overline{Y}]$ containing the equilibrium team performance under any feasible allocation and that there is a unique optimal allocation $\optsharesvector$. Under the residual profits objective, the sum of agents' equity shares under the optimal allocation is increasing in $\beta$, i.e., $$\frac{\partial \left(\sum_{i \in N}\sharesagt{i}^{*}\right)}{\partial \beta} > 0.$$
\end{proposition}

The basic idea behind the proof is that the benefits to retaining more of the firm are linear in the output while the benefits to allocating more shares to workers are convex, and become  steeper as complementarities increase.

If $P(Y)$ is strictly concave, there is a trade-off between the concavity of $P(Y)$ and the convexity of $Y(\actions)$. Depending on which effect is stronger, the fraction of shares allocated to agents may increase or decrease as complementarities grow stronger.

\section{Discussion}

\subsection{The balance condition in more general environments}

Our characterization of the optimal equity allocation relies on several features of our model. In particular, we assume quadratic functional forms for agents' utility and the joint output and assume that heterogeneity across agents arises only from their different network positions. If these assumptions are relaxed, the balanced neighborhood equity result will no longer hold exactly. Nevertheless, the key insight behind the result is more general: optimal incentives favor balancing the spillover effects of incentivizing higher actions.

In general, the principal will trade off the benefits of such balance with other concerns that could be introduced to the model. For example, if agents' individual returns to effort are heterogeneous, a trade-off arises between balancing spillovers and allocating equity to the most individually productive agents. A more complicated balance condition would then be relevant. 

Nevertheless, the forces underlying our main results would remain relevant, and would be the dominant ones in some cases---for instance, in the limit as the spillovers grow large. More precisely, consider extending our model to allow heterogeneous returns $b_i > 0$ to individual effort, so that team performance is
$$\production(\actions) = \sum_{i\in N}b_i\actionagt{i} + \frac{\beta}{2}\sum_{i,j \in N} G_{ij}a_i a_j.$$
Then our main results continue to hold in the limit $\beta \rightarrow \infty$. Intuitively, when complementarities can be sufficiently large, it becomes much more important to exploit those complementarities optimally (which requires balance) than to exploit the heterogeneity in individual productivities.

\subsection{Tightly-knit teams}

Our extensive margin results can be summarized as saying that the principal prefers to concentrate equity in teams whose members have strong mutual complementarities. In terms of interpretation, this need not mean that the agents involved work closely together or are nearby in an organizational sense---just that their efforts are highly complementary in producing the output.

The details of the extensive margin results depend on the specific structure of our model, but we believe the economic intuition underlying these results has broader implications. When the principal motivates an agent by giving him a larger share of the equity in a single common output, it dilutes the equity of the others. Strong complementarities among those getting equity shares countervail this dilution, and this is what makes tight-knit teams valuable to the principal. As we remarked in our discussion of the model's assumptions, in reality a principal may have signals of effort richer than we have studied---for example, outcomes reflecting contributions of specific organizational units. It may be interesting to study how a principal would optimally use multiple signals of this sort to provide incentives to a networked team.

\subsection{A tension between the principal's interests and agent welfare}\label{s:welfare}

\Cref{ex:multiple_active_sets} shows that the principal can be indifferent between active sets, and associated allocations of equity, that have very different welfare implications for workers. In that example, since the cost of effort is convex, agents are better off (on average) if the same performance is achieved by a larger team. But the principal is indifferent between two different team sizes. If the complementarities are perturbed slightly to strengthen those in some maximum clique, then the principal's indifference is broken and she has a strict preference for the compensation scheme that motivates a smaller team---and which happens to leave workers substantially worse off.

This is a consequence of the fact that, subject to paying out a certain total share in equity compensation, the principal is maximizing the probability of a project's success rather than utilitarian welfare. The mechanism of equity pay can do a very poor job of transmitting workers' interest in a more equal distribution of effort. This highlights an interesting tension between welfare and the principal's preferred mode of incentive-provision, and the binary-outcome model we work with brings it out particularly sharply.

\subsection{Connection to a spectral radius maximization problem}\label{s:spectral_radius}

The optimal allocation turns out to have a simple description in terms of a problem of maximizing a spectral radius: if $\boldsymbol{\sigma}$ solves the success probability optimization problem, then  $\boldsymbol{\Sigma}=\operatorname{diag}(\boldsymbol{\sigma})$ also maximizes the spectral radius $\rho(\sharesmatrix \network)$ among nonnegative vectors $\boldsymbol{\sigma}$ summing to $1$.\footnote{The spectral radius of a matrix, which we denote by $\rho(\cdot),$ is the largest magnitude of an eigenvalue of the matrix.} To show this, we show that when $\beta$ is large enough (so that very large spillovers are possible), the principal wants to choose shares $\sharesvector$ inducing a large spectral radius to capture these spillovers. By \autoref{p:invariance}, the optimal allocations do not depend on $\beta$, so in fact such a $\boldsymbol{\sigma}$ is optimal for any $\beta$. We formally state and prove the connection in \autoref{a:spectral_radius}.

An applied mathematics literature discusses spectral radius maximization problems of this form (e.g., \cite{elsner2015spectralmax}, \cite{nesterov2013optimizing}, and \cite*{axtell2009optimization}). Most closely related, \citep{elsner2015spectralmax} consider the same spectral radius maximization problem and discuss algorithms for efficiently computing the optimal diagonal matrix $\sharesmatrix$. Our analysis turns out to provide several insights into this problem. For instance, \autoref{t:unweighted_active_set} implies a characterization of the highest achievable spectral radius when $\network$ is the adjacency matrix of an unweighted network, showing that it is achieved by a dividing shares $\sigma_i$ equally among the members of a maximum clique.

\bibliographystyle{ecta}
{\footnotesize \bibliography{references.bib}}

\appendix

\section{Proofs}
\label{a:appendix_proofs}

\begin{proof}[Proof of \autoref{p:uniqueeq}]
Fixing shares $\sharesvector$ and others' strategies, agent $i$'s expected payoff is strictly concave in his action $a_i$ because $Y(\actions)$ is linear in $a_i$, the success probability $P(Y)$ is concave in $Y$, and the effort cost is strictly convex. So agent $i$ has a unique best response, meaning we need  only consider pure-strategy equilibria. Moreover marginal costs  at $a_i=0$ are zero while marginal benefits at $a_i=0$ are strictly positive if $\sigma_i>0$ and zero if $\sigma_i=0$. Since $U_i$ is concave in $a_i$, this rules out a boundary solution where the first-order condition $\frac{\partial U_i}{\partial a_i}=0$ is not satisfied. So the first-order condition  is necessary and sufficient for a best-response.

It follows that the following equations are necessary and sufficient for the vector $\eqlbactions$ to be a Nash equilibrium: 
$$[\mathbf{I}-P'(Y^*)\beta \sharesmatrix \network]\eqlbactions = P'(Y^*)\sharesvector \text{ and }Y^* = \production(\eqlbactions).$$
For a fixed $Y^*$, there exists a vector $\mathbf{a}^*$ solving the first equation if and only if $P'(Y^*)\beta \rho(\sharesmatrix\network)<1$, where $\rho(\sharesmatrix \network)$ is the spectral radius of $\sharesmatrix \network$. Then equilibrium actions are characterized by
$$\eqlbactions=[\mathbf{I}-P'(Y^*)\beta \sharesmatrix \network]^{-1}P'(Y^*)\sharesvector.$$

Given a constant $y$ such that $P'(y)\beta \rho(\sharesmatrix \network)<1$, we can define actions by $$\eqlbactions(y)=[\mathbf{I}-P'(y)\beta \sharesmatrix \network]^{-1}P'(y)\sharesvector.$$
Solutions of the first-order conditions then correspond to solutions to
$$Y(\eqlbactions(y)) =y.$$

The function $Y(\eqlbactions(y))$ is strictly increasing in each coordinate of $\eqlbactions(y)$ while each coordinate of $\eqlbactions(y)$ is weakly decreasing in $y$ since $P'(\cdot)$ is weakly decreasing (by our assumption $P(\cdot)$ is concave). So  $Y(\eqlbactions(y))$ is decreasing, meaning there is at most one solution to $Y(\eqlbactions(y)) = y$. It remains to show a solution to this equation exists.

We claim that we can find $y$ such that $Y(\eqlbactions(y)) \geq y$ and $P'(y)\beta \rho(\sharesmatrix \network)<1$. If $P'(0)\beta \rho(\sharesmatrix \network)<1$, the claim holds with $y=0$ since $Y(\eqlbactions(0)) \geq 0$. Otherwise, define $y_0$ by $P'(y_0)\beta \rho(\sharesmatrix \network)=1$. A solution to this equation exists since $P'(y)$ is continuous and converges to zero as $y\rightarrow \infty$. Then $Y(\eqlbactions(y)) \rightarrow \infty$ as $y \rightarrow y_0$ from above, so we have $Y(\eqlbactions(y_0+\epsilon)) \geq y_0+\epsilon$ for $\epsilon>0$ sufficiently small. This completes the proof of the claim.

Since $Y(\eqlbactions(y))$ is decreasing in $y$, we can also choose $y$ large enough such that $y > Y(\eqlbactions(y))$. Since $Y(\eqlbactions(y))$ is continuous in $y$, by the intermediate value theorem this function has a fixed point. We conclude that there exists a unique solution to $Y(\eqlbactions(y))  = y$ and a corresponding unique profile $\eqlbactions$ of equilibrium actions.
\end{proof}

\begin{proof}[Proof of \autoref{t:optsharescharacterization}]
\textbf{Part (a)}: Suppose $\sharesvector^*$ is a solution for either objective. Then it must solve the following optimization problem:
\begin{equation}
     \begin{aligned}
     \text{choose $\sharesvector$, $\mathbf{a}$} & \text{ to maximize} & \quad & \production(\mathbf{a}) \\
       & \text{subject to} & & \mathbf{a} = [\mathbf{I}-P'(Y(\mathbf{a}))\beta \sharesmatrix \network]^{-1} P'(Y(\mathbf{a}))\sharesvector \\
       & & & \sharesagt{i} \geq 0 \quad \text{ for all } \, i \in N, \\  & & &   \sum_{i\in N} \sharesagt{i} \leq \sum_{i\in N} \sharesagt{i}^*.
     \end{aligned} \label{problem1}  \end{equation}
The first constraint applies the equilibrium characterization in \Cref{p:uniqueeq}. The content of the final constraint is that if $\sharesvector^*$ is optimal among all feasible allocations, it must maximize equilibrium team performance among allocations distributing at most a fraction $\sum_i \sharesagt{i}^*$ of shares to agents (where this sum will be equal to $1$ under the success probability objective).

Our argument proceeds in several steps.

\medskip

\textit{\textbf{Step 1}: Reduce to the case where all agents are active.}  
We want to show that $(\network \optsharesvector)_i$ is constant across active agents. The values of $(\network \optsharesvector)_i$ for active agents depend only on the restriction of the adjacency matrix to the subnetwork of active agents. Moreover, the allocation $\optsharesvector$ remains optimal for this subnetwork. So without loss of generality, we can drop any inactive agents and assume that all agents in $N$ are active.

\medskip

\textit{\textbf{Step 2}: Reduce to a related problem that avoids the curvature of $P$.} A challenge in studying (\ref{problem1}) is that as we optimize over allocations, the $P'(Y)$ term appearing in the formula for equilibrium actions may change. This step performs a reduction that avoids this complication.

We show any optimal $\sharesvector^*$ must also solve a dual problem of obtaining a given team performance with a minimum fraction of shares distributed to agents.

\begin{lemma}
    
Any solution $(\sharesvector^*,\eqlbactions)$ to (\ref{problem1}) solves the minimization problem
\begin{equation}
     \begin{aligned}
     \text{choose $\sharesvector$, $\mathbf{a}$} & \text{ to minimize} & \quad & \sum_{i\in N} \sigma_i \\
       & \text{subject to} & & Y(\mathbf{a}) = y^* \\
        & & &   \mathbf{a} = [\mathbf{I}-P'(y^*)\beta \sharesmatrix \network]^{-1} P'(y^*)\sharesvector \\
       & & & \sharesagt{i} \geq 0 \quad \text{ for all } \, i \in N.
     \end{aligned} \label{problem2} %
     \end{equation}
\end{lemma}

\begin{proof}
Suppose not, so that there exists a solution $(\sharesvector^*,\eqlbactions)$ to (\ref{problem1}) that does not solve (\ref{problem2}). Fix a solution $(\sharesvector^\dag,\bm{a}^\dag)$ to (\ref{problem2}). By construction $\sharesvector^*$ is feasible for (\ref{problem2}), so we must have \begin{equation}
        \sum_{i \in N}\sigma_{i}^\dag < \sum_{i \in N}\sigma_{i}^{*}. \label{opt1c}
    \end{equation} Since $(\sharesvector^\dag,\bm{a}^\dag)$ is a solution to (\ref{problem2}), it must satisfy its constraints. Thus,  $\production(\bm{a}^\dag) = y^*$ where $\bm{a}^\dag = [\mathbf{I}-\beta P'(y^{*})\sharesmatrix^\dag\network]^{-1}P'(y^{*})\sharesvector^\dag$, which implies that $\sharesvector^\dag$ satisfies the constraints of (\ref{problem1}). So we have a solution to (\ref{problem1}) for which the constraint $\sum_i \sharesagt{i} \leq \sum_i \sharesagt{i}^*$ is a strict inequality. But this gives a contradiction, since the equilibrium team performance is strictly increasing in the shares $\sigma_i$ given to each agent $i$ and so this constraint must bind.
    We conclude that any solution $(\sharesvector^*,\eqlbactions)$ to (\ref{problem1}) is a solution to (\ref{problem2}).
\end{proof}
In the rest of the proof, we will study (\ref{problem2}). This problem is simpler because we have replaced $P'(Y)$ terms with the constant $P'(y^*)$, and therefore need not worry about the curvature of $P$ in our optimization.

\medskip

The third step will calculate a key derivative for studying this problem. First, it will be useful to make some definitions.  The \emph{endogenous Bonacich matrix} at the optimum is defined to be\footnote{By the proof of \autoref{p:uniqueeq}, given an equity allocation $\sharesvector$, the matrix inverse in this expression is well-defined and the equilibrium action profile can be written as $\eqlbactions = \bonacichendo \sharesvector$.}
$$\bonacichendo := P'(y^*)\left[\mathbf{I} - \beta P'(y^*) \sharesmatrix \network \right]^{-1} $$
and the associated \emph{endogenous Boniacich centralities} of the agents are defined as the column sums of this matrix: 
$$b_i :=  \sum_{j \in N} M_{ji}.$$ These are reminiscent of the familiar Bonacich centralities, but as discussed in  \Cref{sec:uniqueeq}, the effective network depends on the principal's choice of $\sharesvector$ and the marginal effect on success probability at the optimum. Finally, define $$\overline{b}_i = \sum_{j \in N} M_{ji}b_{j}$$ to be the weighted sum of $b_j$, where the weights are the corresponding entries in column $i$ of $\mathbf{M}$. This can be thought of as an average of agents' centralities weighted by a measure of their connectedness to agent $i$.

\medskip

\textit{\textbf{Step 3}: A formula for the derivative of team performance in shares.} The purpose of this step is to establish the following formula for the derivative $\frac{dY}{d\sigma_{i}}$ of team performance in the share allocation
$$\frac{dY}{d\sigma_{i}} = \frac{b_{i}\overline{b}_i}{P'(Y)^{2}}.$$

To prove the formula, we begin by using the formula $\production(\eqlbactions) = \boldsymbol{1}^T \eqlbactions + \frac{\beta}{2}(\eqlbactions)^{T} \network\eqlbactions$:
\begin{align}
    \frac{d \production}{d \sigma_{i}} &= \mathbf{1}^{T}\frac{d \eqlbactions}{d \sigma_{i}} + \beta (\eqlbactions)^{T} \network \frac{d \eqlbactions}{d \sigma_{i}} \nonumber \\
    &= [\mathbf{1} + \beta \network \eqlbactions]^{T} \frac{d \eqlbactions}{d \sigma_{i}}. \label{derivativeproduction}
\end{align}

Now we calculate the derivative of equilibrium actions in share allocation, $\frac{d \eqlbactions}{d \sigma_{i}}$.  Rearranging $\eqlbactions = \bonacichendo \sharesvector$, we can characterize the equilibrium action implicitly as follows:
\begin{align}
    \eqlbactions &= \beta P'(y^{*}) \sharesmatrix \network \eqlbactions + P'(y^{*})\sharesmatrix \mathbf{1}, \label{implicitfn} \\
    \iff\sharesmatrix^{-1} \eqlbactions & = \beta P'(y^*) \network \eqlbactions  + P'(y^*)\mathbf{1}. \label{lemma:A}
\end{align}
The elements of the vector $\sharesmatrix^{-1} \eqlbactions $ are equal to $a_i/\sigma_i$, and thus can be interpreted as the average cost, measured in equity shares, of inducing effort by agent $i$. An important fact for the rest of the proof is that this is equal to the endogenous Bonacich centrality of agent $i$.
\begin{lemma} $\sharesmatrix^{-1} \eqlbactions=\beta P'(y^*) \network \eqlbactions  + P'(y^*)\mathbf{1}=\mathbf{b}$. \label{lem:bang_for_stock}\end{lemma} 
\begin{proof}  It will be convenient to work with the transpose:
\begin{align}
    (\sharesmatrix^{-1} \eqlbactions)^T &= (\eqlbactions)^{T}\sharesmatrix^{-1} \nonumber \\
        &= \left([\mathbf{I}-P'(y^*)\beta \sharesmatrix \network ]^{-1}P'(y^*)\sharesmatrix \mathbf{1}\right)^{T}\sharesmatrix^{-1}&& \text{by (\ref{implicitfn})} \nonumber  \\
        & = \mathbf{1}^{T}\sharesmatrix [\mathbf{I}-P'(y^*)\beta  \network \sharesmatrix]^{-1} P'(y^*)\sharesmatrix^{-1}\nonumber  \\
        & = \mathbf{1}^{T} [\mathbf{I}-P'(y^*)\beta \sharesmatrix \network]^{-1}P'(y^*)  && \text{absorbing $\sharesmatrix$, $\sharesmatrix^{-1}$ into the inverse} \label{eq:expression_for_b}  \\
        & =  \mathbf{1}^{T}\bonacichendo \nonumber  && \text{definition of $\bonacichendo$} \\ & = \mathbf{b}^T   \nonumber && \text{definition of $\mathbf{b}$}. && \text{\qedhere}
\end{align}
\end{proof}

Returning to the calculation of $\frac{d \eqlbactions}{d \sigma_{i}}$, we now implicitly differentiate  (\ref{implicitfn}) in $\sigma_i$:
\begin{align}
    \frac{d \eqlbactions}{d \sigma_{i}} &= \beta \frac{d \sharesmatrix}{d \sigma_{i}}\network \eqlbactions P'(y^*) + \beta \sharesmatrix \network \frac{d \eqlbactions}{d \sigma_{i}}P'(y^*) + \frac{d \sharesmatrix}{d \sigma_{i}}\mathbf{1}P'(y^*). \nonumber
\end{align}
Simplifying the above expression, we obtain
\begin{align}
    \left[\mathbf{I}-\beta \sharesmatrix \network P'(y^*)\right] \frac{d \eqlbactions}{d \sigma_{i}} &= P'(y^*) \begin{bmatrix}
                 0  \\
                 1+\beta\left(\network \eqlbactions\right)_{i} \\
                 0
            \end{bmatrix}, \nonumber
\end{align}
where the nonzero entry of the last vector is in row $i$. Premultiplying this expression by $\left[\mathbf{I}-\beta \sharesmatrix \network P'(y^*)\right]^{-1}$, we get 
\begin{align}
    \frac{d \eqlbactions}{d \sigma_{i}} &= \left[\mathbf{I}-\beta \sharesmatrix \network P'(y^*)\right]^{-1}P'(y^*) \begin{bmatrix}
                 0  \\
                 1+ \beta \left(\network \eqlbactions\right)_{i} \\
                 0
            \end{bmatrix} \nonumber \\
            &= \mathbf{M} \begin{bmatrix}
                 0  \\
                 1+ \beta \left(\network \eqlbactions\right)_{i} \\
                 0
            \end{bmatrix} . \label{derivativeeqlbactions}
\end{align} 

Now, substituting (\ref{derivativeeqlbactions}) in (\ref{derivativeproduction}), we deduce
$$
    \frac{d \production}{d \sigma_{i}} =  \left[\mathbf{1} + \beta \network \eqlbactions\right]^{T}\mathbf{M} \begin{bmatrix}
                 0  \\
                 1+ \beta \left(\network \eqlbactions\right)_{i} \\
                 0
            \end{bmatrix}.$$Using the formula of \Cref{lem:bang_for_stock} twice and the definition of $\overline{b}_i$,  we conclude
$$\frac{dY}{d\sigma_{i}} = \frac{b_{i}\overline{b}_i}{P'(Y)^{2}}.$$

\medskip

\textit{\textbf{Step 4}: Deduce that the endogenous Bonacich centralities $b_i$ are constant across agents.} We know $(\optsharesvector,\mathbf{a}^*)$ is a solution to (\ref{problem2}).  We can combine the two constraints $Y(\mathbf{a}) = y^*$ and $\mathbf{a} = [\mathbf{I}-P'(y^*)\beta \sharesmatrix \network]^{-1} P'(y^*) \sharesvector$ into a single constraint, \begin{equation} Y([\mathbf{I}-P'(y^*)\beta \sharesmatrix \network]^{-1} P'(y^*) \sharesvector) = y^*,\label{eq:lumped_constraint}\end{equation} noting that  $P'(y^*)$ is now a constant, so the argument of the function $Y$ here depends only on $\bm{\sigma}$. The Lagrangian first-order conditions for the share-minimization problem then  read $$ \frac{d \sum_i \sigma_i}{d \sigma_i} - \lambda \frac{d Y}{d \sigma_i}= 0, $$ where $\lambda$ is the multiplier on the constraint (\ref{eq:lumped_constraint}) we have just introduced. (Recall we have dropped any inactive agents, so the nonnegativity constraints do not bind.) Since $\frac{d \sum_i \sigma_i}{d \sigma_i} =1$ for all agents, optimality requires that $\frac{dY}{d\sigma_i}$ is equal to some constant for all agents $i$. This is intuitive: it should not be possible to increase output by transferring shares from one agent to another.  By the formula obtained in Step 3, this requires that
$$b_i \overline{b}_i \text{ is constant across agents $i$} .$$ 

The idea of the rest of this step is to show that this is possible only if all the $b_i$ are identical. The intuition is that both $b_i$ and $\overline{b}_i$ are centrality measures closely related to Bonacich centrality.  If the $b_i$ are not all identical, then the agent with maximum $b_i$ must have minimum $\overline{b}_i$, a strange situation. To show formally this is impossible, suppose that there exist two agents $i^{*} \in N$ with $i^{*} = \argmin_{k \in N}b_{k}$ and $j^{*} \in N$ with $j^{*} = \argmax_{k \in N} b_{k}$ such that $b_{i^{*}} < b_{j^{*}}$.\footnote{We are grateful to Michael Ostrovsky for suggesting the argument in the next paragraph.} 

Then we have that, for agent $i^{*}$, $$b_{i^{*}}\overline{b}_{i^*} < b_{i^{*}}b_{j^{*}}\sum_{j \in N}M_{ji^{*}} = (b_{i^{*}})^{2}b_{j^{*}},$$ using the maximality of $b_{j^*}$ among the $b_j$ and  the definitions of $\overline{b}_i$ and $b_i$. But we similarly have that, for agent $j^{*}$, $$b_{j^{*}}\overline{b}_{j^*} > b_{j^{*}}b_{i^{*}}\sum_{j \in N}M_{jj^{*}} = b_{i^{*}}(b_{j^{*}})^{2}.$$ 
 We showed above that $b_{i^{*}}\overline{b}_{i^*}=b_{j^{*}}\overline{b}_{j^*}$ for any two agents $i^*$ and $j^*$, and so combining the last two inequalities implies
 $$(b_{i^*})^2b_{j^*} > b_{i^*}(b_{j^*})^2.$$
This contradicts our assumption $b_{j^*}>b_{i^*}$, so we must have $b_i$ equal to some constant $c_1$ for all $i$ in $N$.

\medskip

\textit{\textbf{Step 5}: Conclude the neighborhood equity $(\network \sharesvector)_i$ is constant across (active) agents.} By  (\ref{eq:expression_for_b}) and Step 2, we have
$$ (\mathbf{1}^{T} [\mathbf{I}-P'(y^*)\beta \sharesmatrix \network]^{-1}P'(y^*))_i = c_1$$
for all $i$. Therefore,
$$P'(y^*) = c_1-P'(y^*)\beta  c_1 (\network\sharesvector)_i $$
for all  $i$, so there exists a constant $c$ such that $(\network\sharesvector)_i=c$ for all $i$ (among the subnetwork of active agents).

\bigskip

\textbf{Part (b)}: We now compute the equilibrium action vector under the optimal allocation. We calculate
\begin{align*}
        \eqlbactions &= \left[\mathbf{I}-P'(Y^*)\beta \optsharesmatrix \network \right]^{-1} P'(Y^*)\optsharesvector, \\
        &= \sum_{k=0}^{\infty}P'(Y^*)^k\beta^{k}\left(\optsharesmatrix \network\right)^{k}P'(Y^*) \optsharesvector, \\
        &= \sum_{k=0}^{\infty}P'(Y^*)^k\beta^{k}c^{k}\optsharesvector \\
        &= \frac{P'(Y^*)\optsharesvector}{1-P'(Y^*)\beta c},
\end{align*}
where the first equality follows from writing out the expansion of the inverse. The third equality follows from the characterization of the optimal equity shares above for entries corresponding to active agents, and holds for inactive agents because both sides are zero. Convergence of the geometric series follows from the proof of \autoref{p:uniqueeq}.

\bigskip

\textbf{Part (c)}: For all active agents $i$, we have
$$(\network \eqlbactions )_i = \frac{P'(Y^*)(\network \optsharesvector)_i}{1-P'(Y^*)\beta c} = \frac{P'(Y^*)c}{1-P'(Y^*)\beta c}$$
where the first equality applies part (b) and the second equality applies part (a).
\end{proof}

We next state and prove a lemma, discussed in \Cref{sec:active}, which will be used in subsequent proofs.

\begin{lemma}
\label{l:activesetopt}
An allocation $\sharesvector$ is optimal among those with a given sum of shares $\sum_{i\in N} \sigma_i = s$ if and only if it solves \begin{equation*}  \begin{aligned}
    & \max_{\sharesvector} && c \\
    & \text{subject to} && (\network \boldsymbol{\sigma})_i = c \text{ whenever }\sigma_i > 0.
  \end{aligned}
\end{equation*}
In particular, the active sets at optimal allocations are the same as the sets of non-negative indices under solutions to this optimization problem.
\end{lemma}

\begin{proof}
   By \autoref{t:optsharescharacterization}, there exists a constant $c$ such that for all agents that get a strictly positive equity at the optimal solution, $(G\boldsymbol{\sigma}^{*})_{i}=c$. At any solution which satisfies the balanced equity condition and allocates a fraction $s$ of shares to agents, the team performance  $$Y^{*} = \mathbf{1}^{T}\eqlbactions + \frac{\beta}{2}(\eqlbactions)^{T}\mathbf{G}\eqlbactions$$ can be rewritten as 
   
\begin{equation}\label{eq:team_performance_expanded}
       Y^* = \left(\frac{P'(Y^{*})}{1-\beta P'(Y^{*}) c} + \frac{\beta P'(Y^{*})^{2} c}{2(1-\beta P'(Y^{*}) c)^{2}}\right)s.
\end{equation}
   
  We will conclude from the above expression that team performance is increasing in $c$. For a given $c$, the team performance $Y^{*}(c)$ is the solution to $f(y,c)=y$, where we define $$f(y,c) :=  \left(\frac{P'(y)}{1-\beta P'(y) c} + \frac{\beta P'(y)^{2} c}{2(1-\beta P'(y) c)^{2}}\right)s.$$ By assumption $P(\cdot)$ is concave and twice differentiable, so $f(y,c)$ is decreasing in $y$. Since we have also assumed $P(\cdot)$ is strictly increasing, we have $\frac{\partial f}{\partial c} > 0$ for all $y$ and thus $Y^{*}(c)$ is increasing in $c$.
\end{proof}

\begin{proof}[Proof of \autoref{p:invariance}]
By \autoref{l:activesetopt}, the optimal active set is determined by an optimization problem that is independent of $\beta$. So if the active set for the allocation $\sharesvector_0^*$ is optimal for complementarity parameter $\beta_0$, the same active set is optimal for complementarity parameter $\beta_1$. Call this active set $\subnetwork$.

By \autoref{t:optsharescharacterization}, the optimal allocation with this active set is characterized by $\optsharesvector_0 = \subnetwork^{-1} \mathbf{1} c_0$ for some constant $c_0>0$ when the complementarity parameter is $\beta_0$. Similarly, the allocation with this active set is characterized by $\optsharesvector_1 = \subnetwork^{-1} \mathbf{1} c_1$ for some constant $c_1>0$ when the complementarity parameter is $\beta_0$. This implies the claim with $k = c_1/c_0$.
\end{proof}

\begin{proof}[Proof of \autoref{p:diameter}]
We consider the success probability objective as the argument is essentially the same for both objectives. By \autoref{l:activesetopt}, any optimal allocation maximizes $(G\sharesvector)_i$ for active agents $i$ among allocations $\sharesvector$ satisfying the balanced equity condition.

Let $\overline{g}=\max_{i,j} G_{ij}$ and choose $i$ and $j$ such that the link between $i$ and $j$ obtains this maximum weight. Setting $\sigma_i=\sigma_j=\frac12$ and all other $\sigma_k=0$ gives $(G\sharesvector)_i=(G\sharesvector)_j = \overline{g}/2$. We now show this value cannot be obtained with an active set with diameter greater than $2$.

Suppose there is an optimal allocation $\sharesvector^*$ with an active set $A^*$ with diameter greater than $2$. Choose active agents $i$ and $j$ such that the distance between $i$ and $j$ is at least $2$. The subsets $\{i\},\{j\},N(i) \cap A^*,$ and $N(j) \cap A^*$ of the active set are all disjoint.

The balanced equity condition implies that $(G\sharesvector^*)_i=(G\sharesvector^*)_j =c$ for some constant $c$, and we have
\begin{align*}2c&=(G\sharesvector^*)_i + (G\sharesvector^*)_j\\& \leq \overline{g}  \sum_{k \in N(i) \cup N(j)} \sigma^*_k \\ &<  \overline{g} ,\end{align*}
where the last inequality holds because $\sigma^*_i>0$ so $\sum_{k \in N(i) \cup N(j)} \sigma^*_k <1$. Since we showed we can obtain a value of $c=\overline{g}/2$, this contradicts the optimality of $\optsharesvector$.
\end{proof}

\begin{proof}[Proof of \autoref{t:unweighted_active_set}]
   We consider the success probability objective as the argument is essentially the same for both objectives.  By \autoref{l:activesetopt}, any optimal allocation maximizes the constant $c=(G\sharesvector)_i$ for active agents $i$ among allocations $\sharesvector$ satisfying the balanced equity condition.
   
   Let the size of the maximum clique in the network be $\overline{k}$. \autoref{t:optsharescharacterization} implies that the optimal allocation with active set a clique of size $k$ gives all active agents equal shares $\frac{1}{\overline{k}}$. Under this allocation, the balanced equity condition holds with constant $c=\frac{\overline{k}-1}{\overline{k}}$.
    
    Suppose an allocation $\sharesvector$ satisfies the balanced neighborhood equity condition with constant $c > \frac{\overline{k}-1}{\overline{k}}$. We will show the active set under this allocation must contain a clique of size at least $\overline{k}+1$, which contradicts our assumption that the size of the maximum clique is $\overline{k}$. 
    
    Define $$k^{*} = \argmax_{k \in \mathbb{Z}}\left\{c - \left(\frac{k-1}{k}\right) > 0 \right\}.$$
    We  will show that the active set under $\sharesvector$ contains a clique of size $k^{*}+1 > \overline{k}$. Call the set of vertices of this active set by $A^*$. First observe that the equity that each agent gets is at most $(1-c)$. This is because each agent's neighbors receive equity shares summing to $c$ and the total of all equity shares is $1$.

    We will define a sequence of agents $i^0,\hdots,i^{k^*}$ inductively such that $i^0,i^1,\hdots,i^k$ is a clique for all $k$. Fix some $i^0$ in the active set and define  $N_{S}(i^{0}) := N(i^{0})$. Given $i^0,i^1,\hdots,i^k$ for any $k< k^*$, we define $$N_{S}(i^{k}) :=\bigcap_{l=0}^{k}N(i^{l}).$$
    
    Given $i^0,\hdots,i^k$ with $0 \leq k < k^*$, we want to choose $i^{k+1}$ to be an arbitrary agent in $N_S(i^k)$. To do so, we must show $N_S(i^k)$ is non-empty.
    
    We will prove that the total equity in $N_{S}(i^{k})$ is at least $(k+1)c - k$, i.e., $$\sum_{i \in N_{S}(i^{k})}\sigma_i \geq (k+1)c - k.$$ We show this by induction on $k$. The base case $k=0$ holds by the balanced neighborhood equity condition.

    The inductive hypothesis is 
$$\sum_{i \in N_{S}(i^{k-1})}\sigma_i \geq kc - (k-1).$$
This implies \begin{equation}\label{e:inductionmodified}\sum_{i \in N(i^{0}) \setminus N_{S}(i^{k-1})}\sigma_i \leq (k-1)(1-c)\end{equation}
since $\sum_{i \in N(i^{0})}\sigma_i = c$.

Since $i^{k}$ is active, we also have $\sum_{i \in N(i^{k})}\sigma_i = c$. We can decompose the equity in this neighborhood, potentially along with  additional agents' shares, as
    $$
        \sum_{i \in A^*\setminus N(i^{0})}\sigma_i + \sum_{i \in N(i^{0}) \setminus N_{S}(i^{k-1})}\sigma_i + \sum_{i \in N_{S}(i^{k})}\sigma_i \geq c.$$
    By the balanced neighborhood equity condition for agent $i^0$ and (\ref{e:inductionmodified}), this implies
        $$ \sum_{i \in N_{S}(i^{k})}\sigma_i \geq (k+1)c - k,$$
completing the induction. Since $c > \frac{k^{*}-1}{k^{*}}$, this implies that $N_{S}(i^{k})$ is non-empty for each $k \in \{1,\dots,k^{*}-1\}$. So we can construct $i^0,\hdots,i^{k^*}$ as described above.

By construction, the subnetwork  $\{i^0,\hdots,i^{k^*}\}$ is a clique of size $k^*+1$. Since we have assumed the maximum clique has size $\overline{k}$, this contradicts the existence of an allocation $\sharesvector$ satisfying the balanced neighborhood equity condition with constant $c > \frac{\overline{k}-1}{\overline{k}}$. Thus the maximum clique must be an optimal allocation.
\end{proof}

\begin{proof}[Proof of \autoref{p:compstaticmatrixcalc}]
    \autoref{t:optsharescharacterization} tells us that, for all agents such that $\optsharesagt{i} > 0$, we have
$$\optsharesvector = c\subnetwork^{-1} \textbf{1}.$$
We will use the matrix calculus expression 
$$\frac{\partial \network(t)^{-1}}{\partial t} = - \network(t)^{-1} \frac{\partial \network(t)}{\partial t} \network(t)^{-1}.$$ Taking the derivative with respect to $G_{jk}$, we have that

$$\frac{\partial \optsharesvector}{\partial G_{jk}} =  - c \subnetwork^{-1} \frac{\partial \subnetwork}{\partial G_{jk}} \subnetwork^{-1} \mathbf{1}+\frac{\partial c}{\partial G_{jk}} \subnetwork^{-1} \textbf{1}.$$
Analyzing the $i^{th}$ element in this vector gives

$$\frac{\partial \sharesagt{i}^*}{\partial G_{jk}} = -c(\subnetwork^{-1}\mathbf{1})_j (\subnetwork^{-1})_{ik} -c(\subnetwork^{-1}\mathbf{1})_k (\subnetwork^{-1})_{ij}+ \frac{\partial c}{\partial G_{jk}}\cdot (\subnetwork^{-1}\mathbf{1})_i. \eqno  $$
The result follows from $\sigma_i^* = c(\subnetwork^{-1}\mathbf{1})_i$ and the analogous expressions with indices $j$ and $k$.\end{proof}

\begin{proof}[Proof of \autoref{p:compstaticoveralloutput}]
    We want to calculate the derivative of the team performance $Y^*$ under the optimal allocation as $G_{ij}$ increases. By the envelope theorem, we can calculate this derivative by holding fixed the allocation $\optsharesvector$. To do so, we calculate the derivative of the equilibrium team performance $Y^*$ for a given allocation $\boldsymbol{\sigma}$ as $G_{ij}$ increases. We will then substitute $\boldsymbol{\sigma}=\optsharesvector$.
    
    Letting $\mathbf{a}^{*}$ be the equilibrium action profile under the allocation $\boldsymbol{\sigma}$, we calculate
    \begin{equation*}
        \begin{split}
            \frac{\partial Y}{\partial G_{ij}} &= \frac{\partial \mathbf{1}^{T}\mathbf{a^{*}} + \frac{\beta}{2}(\mathbf{a}^{*})^{T}\mathbf{G}\mathbf{a}^{*}}{\partial G_{ij}}, \\
            & = \mathbf{1}^{T}\frac{\partial \mathbf{a^{*}}}{\partial G_{ij}} + \beta (\mathbf{a}^{*})^{T}\mathbf{G}\frac{\partial \mathbf{a^{*}}}{\partial G_{ij}} + \frac{\beta}{2}(\mathbf{a}^{*})^{T} \frac{\partial \mathbf{G}}{\partial G_{ij}}\mathbf{a}^{*}, \\
            & = \left[\mathbf{1}^{T} + \beta (\mathbf{a}^{*})^{T}\mathbf{G}\right]\frac{\partial \mathbf{a^{*}}}{\partial G_{ij}} + \beta a_{i}^{*}a_{j}^{*}.
        \end{split}
    \end{equation*}
    The equilibrium action satisfies $\mathbf{a}^{*} = \beta P'(Y)\boldsymbol{\Sigma} \mathbf{G} \mathbf{a}^{*} + P'(Y)\boldsymbol{\Sigma}\mathbf{1}$. Thus, we can write
    \begin{equation*}
        \begin{split}
            \frac{\partial \mathbf{a^{*}}}{\partial G_{ij}} &= \beta P'(Y)\boldsymbol{\Sigma}\frac{\partial \mathbf{G}}{\partial G_{ij}}\mathbf{a}^{*} + \beta P'(Y) \sharesmatrix \network \frac{\partial \mathbf{a^{*}}}{\partial G_{ij}} + \left(\beta \boldsymbol{\Sigma} \mathbf{G} \mathbf{a^{*}} + \boldsymbol{\Sigma}\mathbf{1}\right)\frac{\partial P'(Y)}{\partial G_{ij}}, \\
            & = \beta \begin{bmatrix}
                 0  \\
                 \sigma_{i}a_{j}^{*} \\
                 0 \\
                 \sigma_{j}a_{i}^{*} \\
                 0
            \end{bmatrix}P'(Y) + \beta P'(Y) \sharesmatrix \network \frac{\partial \mathbf{a^{*}}}{\partial G_{ij}} + \left(\beta \boldsymbol{\Sigma} \mathbf{G} \mathbf{a^{*}} + \boldsymbol{\Sigma}\mathbf{1}\right)P''(Y)\frac{\partial Y}{\partial G_{ij}}.
        \end{split}
    \end{equation*}
    where $\boldsymbol{\Sigma}\frac{\partial \mathbf{G}}{\partial G_{ij}}\mathbf{a}^{*}$ is a vector with the $i^{th}$ element equal to $\sigma_{i}a_{j}^{*}$, the $j^{th}$ element equal to $\sigma_{j}a_{i}^{*}$ and the rest of the elements equal to zero. Solving for $\frac{\partial \mathbf{a^{*}}}{\partial G_{ij}} $ gives $$ \frac{\partial \mathbf{a^{*}}}{\partial G_{ij}} = \left[\mathbf{I}-\beta P'(Y)\boldsymbol{\Sigma}\mathbf{G}\right]^{-1}\left(\beta P'(Y) \begin{bmatrix}
                 0  \\
                 \sigma_{i}a_{j}^{*} \\
                 0 \\
                 \sigma_{j}a_{i}^{*} \\
                 0
            \end{bmatrix} + \boldsymbol{\Sigma}\left[\mathbf{1} + \beta \mathbf{G} \mathbf{a}^{*}\right]P''(Y)\frac{\partial Y}{\partial G_{ij}}\right).$$
Substituting into the expression for $\frac{\partial Y}{\partial G_{ij}}$ gives
\begin{multline*}
    \frac{\partial Y}{\partial G_{ij}} \left[1-\left(\mathbf{1}+\beta  \mathbf{G} \mathbf{a}^{*}\right)^{T}\left[\mathbf{I}-\beta P'(Y)\boldsymbol{\Sigma} \mathbf{G}\right]^{-1}\boldsymbol{\Sigma} \left(\mathbf{1}+\beta  \mathbf{G} \mathbf{a}^{*}\right)P''(Y)\right]  \\ = \beta P'(Y)\left[\mathbf{1}^{T} + \beta (\mathbf{a}^{*})^{T}\mathbf{G}\right]\left[I - \beta P'(Y) \sharesmatrix \network \right]^{-1}\begin{bmatrix}
                 0  \\
                 \sigma_{i}a_{j}^{*} \\
                 0 \\
                 \sigma_{j}a_{i}^{*} \\
                 0
            \end{bmatrix} + \beta a_{i}^{*}a_{j}^{*}.
\end{multline*}

We now use the optimality of $\boldsymbol{\sigma}$, which implies the equality $\mathbf{a}^{*} = \optsharesvector\frac{P'(Y)}{1-\beta c P'(Y) }$ by \autoref{t:optsharescharacterization}. Applying this, we obtain
$$\frac{\partial Y^*}{\partial G_{ij}} = \beta \sigma_{i}^{*}\sigma_{j}^{*}P'(Y^*)^{2}\frac{\left(\frac{2}{(1-\beta c P'(Y^*))^{3}} + \frac{1}{(1-\beta c P'(Y^*))^{2}}\right)}{1-\frac{P''(Y^*) \sum_{i}\sigma^{*}_{i}}{(1-\beta c P'(Y^*))^{3}}}.$$
The right-hand side has the desired form.
\end{proof}

\begin{proof}[Proof of \autoref{p:compstatictotalequity}]

  We can assume without loss of generality that all agents in $\network$ are active under $\optsharesvector$ (by dropping any inactive agents from the network). Consider a feasible allocation $\sharesvector$ satisfying the balanced neighborhood equity condition $\mathbf{G}\boldsymbol{\sigma} = c\mathbf{1}$ and let $s=\sum_{i}\sigma_{i} \in [0,1]$ be the sum of shares under this allocation. Such an allocation will exist for any $s \in [0,1]$, as $\mathbf{G}$ is the optimal active set and thus $\left(\mathbf{G}^{-1}\mathbf{1}\right)_{i} > 0$ for all $i$ in $\network$. The balanced neighborhood equity condition implies that $$c = \frac{s}{\mathbf{1}^{T}\mathbf{G}^{-1}\mathbf{1}}.$$
   
   Applying (\ref{eq:team_performance_expanded}), we can write the residual profit for the principal under this allocation as
   $$V(s,\beta) = \alpha^{2} s(1-s)\left(\frac{1}{1-\beta \alpha \frac{s}{\mathbf{1}^{T}\mathbf{G}^{-1}\mathbf{1}}} + \frac{\beta \alpha \frac{s}{\mathbf{1}^{T}\mathbf{G}^{-1}\mathbf{1}}}{2\left(1-\beta \alpha \frac{s}{\mathbf{1}^{T}\mathbf{G}^{-1}\mathbf{1}}\right)^{2}}\right).$$ So for a fixed $\beta$ for which $\optsharesvector$ is an optimal allocation, the sum of shares under this allocation solves the optimization problem $$V^{*}(\beta) = \max_{s \in [0,1]}s(1-s)\left(\frac{1}{1-\beta \alpha \frac{s}{\mathbf{1}^{T}\mathbf{G}^{-1}\mathbf{1}}} + \frac{\beta \alpha \frac{s}{\mathbf{1}^{T}\mathbf{G}^{-1}\mathbf{1}}}{2\left(1-\beta \alpha \frac{s}{\mathbf{1}^{T}\mathbf{G}^{-1}\mathbf{1}}\right)^{2}}\right).$$
   
   We will characterize the solution to this optimization problem. We define $k^{*}:=\mathbf{1}^{T}\mathbf{G}^{-1}\mathbf{1}$ and claim that  $\beta \alpha < k^{*}$. We must have $\beta \in \left(0,\frac{1}{\alpha}\frac{1}{\rho(\mathbf{\Sigma G})}\right)$ by our assumption that equilibrium team performance is in $[0,\overline{Y}]$. Observe that $c = s/k^*$ is an eigenvalue for the matrix $\mathbf{\Sigma G}$ for any $s \in [0,1]$, with right eigenvector $\sharesvector$. Thus we have $$\frac{s}{k^{*}} \leq \rho(\mathbf{\Sigma G}) < \frac{1}{\beta \alpha}.$$ Choosing $s=1$ then verifies the claim $\beta \alpha < k^{*}$. 
   
   We now return to the problem of maximizing $V(s,\beta)$. Taking the partial derivative with respect to $s$, we find $$\frac{\partial V(s,\beta)}{\partial s} = \frac{k^{*}\alpha^{2}\left(-(\beta \alpha)^2 s^{3} + 3 \beta \alpha k^{*} s^2 - 4(k^{*})^{2}s + 2(k^{*})^{2}\right)}{2\left(k^{*}-\beta \alpha s\right)^{3}}.$$
   It suffices to study the behavior of $V(s,\beta)$ when $s\in [0,1]$. Define $$p(s,\beta):=-(\beta \alpha)^2 s^{3} + 3 \beta \alpha  k^{*} s^2 - 4(k^{*})^{2}s + 2(k^{*})^{2}$$ to be the numerator of $V(s,\beta)$. The partial derivative of $p(s,\beta)$ with respect to $s$ is $$\frac{\partial p(s,\beta)}{\partial s} = -3(\beta \alpha)^{2}s^{2} + 6 \beta \alpha k^{*} s - 4(k^{*})^{2} < -3(\beta \alpha s + k^*)^2.$$
   Since the right-hand expression is strictly negative, the function $p(s,\beta)$ is strictly decreasing in $s \in [0,1]$. Thus $p(\cdot,\beta)$ has only one real root for each $\beta$.
   
   We claim that this root lies in $\left(\frac{1}{2},1\right)$. At $s=\frac{1}{2}$, we have $$p\left(\frac{1}{2},\beta\right) = \left(-(\beta \alpha)^2 \cdot \frac{1}{8} + 3\alpha \beta k^{*} \cdot \frac{1}{4}\right) > 0,$$ for any $\beta \alpha < k^{*}$. At $s=1$, we have $$p(1,\beta) = (k^{*}-\beta \alpha)(\beta \alpha - 2k^{*}) < 0,$$ for any $\beta \alpha < k^{*}$. This proves the claim.
    
   For $s \in [0,1]$, the denominator of $V(s,\beta)$ is strictly positive for any $\beta \alpha < k^{*}$. So for each $\beta$, the sum of shares $s$ at the optimal allocation is characterized by $p(s,\beta)=0$. We calculate$$\frac{\partial p(s,\beta)}{\partial \beta} = 3\alpha k^{*}s^{2} - 2\beta \alpha^{2}s^{3} = \alpha s^{2}\left(3k^{*}-2\beta \alpha s\right) > 0,$$ where the inequality holds for all $s \in (0,1)$. Since $p(s,\beta)$ is strictly decreasing in $s$ for each $\beta$, the sum of shares $s$ at the optimal allocation is increasing in $\beta$.
\end{proof}

\section{Three-agent networks}
\label{a:three_agent}

In this section, we describe the optimal equity allocation in a three-agent weighted network (without self-links). We begin by describing the optimal allocation and then discuss how this allocation changes as we vary the network.

We first describe which agents receive the most equity and which agents are active as a function of the network structure. We then derive comparative statics of the optimal equity shares as the network structure varies. We will see that the equity given to an agent need not be monotonic with respect to the weights on his links.

Without loss of generality, we can assume the adjacency matrix is
$$\network = 
\begin{bmatrix}0 & 1 & G_{13} \\
1 & 0 & G_{23} \\
G_{13} & G_{23} & 0
\end{bmatrix},$$

\noindent
with the ordering $1 \geq G_{13} \geq G_{23} \geq 0$.\footnote{To reduce to this case, we can reorder agents and then normalize the link $G_{12}$ to $1$.} We will consider the case when all three links have positive weight, i.e., $G_{23}>0$.

Our first result ranks agents' equity allocations and describes the active set:
\begin{proposition}
\label{p:threeagentoptshareordering}
For any $\beta > 0$, there is a unique optimal allocation. The set of active agents in an optimal allocation is $\{1,2,3\}$ if and only if
$$G_{13} + G_{23} > 1.$$
The shares at the optimal allocation are ranked $$\optsharesagt{1} \geq \optsharesagt{2} \geq \optsharesagt{3}.$$
\end{proposition}

We defer proofs to the end of this section. The first part of the proposition states the optimal allocation is unique, or equivalently there is a unique best active set for all networks. The second part of the proposition characterizes the active set under the optimal allocation. All agents are active precisely when each agent's degree is greater than the weight of the opposing link, and otherwise only the two higher-degree agents are active. Intuitively, it must be the case that the combination of the two lower-weight links provides complementarities at least as strong  as the highest-weight link. Otherwise, the principal prefers to give equity only to the two higher-degree agents.

The final part of the proposition describes which agents receive the most equity. Under any optimal allocation, equity shares are increasing in degree. Intuitively, the weights on links originating from agent $1$  dominate those originating from agent $2$, which in turn dominate links from agent $3$. From the balanced neighborhood equity characterization of the optimal equity incentives, we expect that agents with stronger links receive more equity.\footnote{Roughly speaking, an agent $i$'s degree times average equity in $i$'s neighborhood is equal to a constant. So agents with higher degree need to have neighbors with lower equity. With three agents this means that higher degree implies higher equity.}

We now turn to comparative statics of the optimal equity allocation with respect to the weight of the weakest link $G_{23}$, fixing the weight $G_{13} > \frac12$. By \autoref{p:threeagentoptshareordering}, agent $3$ is active if and only if $G_{23}>1-G_{13}$. We will restrict to this interval to take comparative statics, as the optimal allocation does not depend on $G_{23}$ below the threshold $1-G_{13}$.

\begin{proposition}
        \label{p:compstatedgeweight}

Fix $G_{13} > \frac12$. When $G_{23} \in (1-G_{13},G_{13})$, we have
$$\frac{\partial \left(\frac{\optsharesagt{2}}{\optsharesagt{3}}\right)}{\partial G_{23}} < 0, \quad \text{ and } \quad \frac{\partial \left(\frac{\optsharesagt{1}}{\optsharesagt{3}}\right)}{\partial G_{23}} < 0. $$

Moreover, there exists a threshold $g^* \in (0,1)$, which is a function of $G_{13}$, such that
$$\frac{\partial \left(\frac{\optsharesagt{1}}{\optsharesagt{2}}\right)}{\partial G_{23}}  
  >0    \mbox{ if } G_{23} < g^*(G_{13}) \quad \text { and } \quad
  \frac{\partial \left(\frac{\optsharesagt{1}}{\optsharesagt{2}}\right)}{\partial G_{23}}   <0   \mbox{ if } G_{23} > g^*(G_{13})
.$$
\end{proposition}

Under the success probability objective, all equity is allocated and there is no  question of how the sum of shares changes; in that case, an immediate implication of \autoref{p:compstatedgeweight} is that the equity of agent $2$ is decreasing in $G_{23}$ for small $G_{23}$, while the equity of agent $1$ is decreasing in $G_{23}$ for large $G_{23}$.

\begin{corollary}
\label{c:threeagtequityredn}
In the success probability problem, increasing $G_{23}$ decreases agent $2$'s optimal share $\sigma^*_2$  when all agents are active and $G_{23} < g^*(G_{13})$. Increasing  $G_{23}$ decreases agent $1$'s optimal share $\sigma^*_1$  when $G_{23} > g^*(G_{13})$. 
\end{corollary}

The first part of the result is more counterintuitive: it says that increasing the strength of an agent's collaborations can decrease the agent's relative equity share. This is what happens to agent $2$ when $G_{23}$ is increased starting from a low level. Once $G_{23}$ gets large enough, however, the equity given to agent $2$ begins increasing; this is financed by taking equity from agent $1$, whose stake is then decreasing in this regime by the second part of the corollary.

\subsection{Proofs for three-agent networks}

We now prove the two propositions in this section.

\begin{proof}[Proof of \autoref{p:threeagentoptshareordering}]
As a preliminary, we prove any optimal active set has at least two agents. Suppose only agent $i$ were active. Let the optimal equity allocation for such an agent be $\sigma^{*}$. The principal then maximizes the output $(1-\sigma)P(\sigma)$, so $\sigma^{*}$ solves $P(\sigma^{*}) = (1-\sigma^{*})P'(\sigma^{*})$ for the residual profits objective or $\sigma^{*} = 1$ for the success probability objective. But, for any $\beta > 0$, the principal can strictly increase residual profits by giving agents $1$ and $2$ each shares $\sigma_1=\sigma_2=\sigma^{*}/2$ or increase the success probability by giving shares $\sigma_1=\sigma_2=\frac12$, which gives a contradiction.

We claim that the unique active set is all three agents if and only if $G_{13} + G_{23} > 1$. If all agents are active at an optimal allocation, then \autoref{t:optsharescharacterization} characterizes the optimal shares as
$$\optsharesagt{1} = \frac{1+G_{13} - G_{23}}{2G_{13}} c, \quad\optsharesagt{2} = \frac{1+G_{23} - G_{13}}{2G_{23}} c, \quad \optsharesagt{3} = \frac{G_{23} + G_{13} - 1}{2G_{23}G_{13}} c,$$
for some constant $c$. A necessary condition for $\{1,2,3\}$ to be an optimal active set is $\sigma_3^*>0$, which is equivalent to $G_{13} + G_{23} > 1$. 

To prove the claim, it remains to show this condition is sufficient: if $G_{13} + G_{23} > 1$, then the unique active set is all three agents. Fix an optimal allocation $\optsharesvector$ and let $s = \sum_i \sigma_i^*$ be the corresponding sum of shares. By \autoref{l:activesetopt}, this optimal allocation must maximize the constant $c$ in the expression $\subnetwork \sharesvector = c \mathbf{1}$ among feasible allocations with $\sum_i \sigma_i= s$.

First suppose the active set is $\{1,2\}$. By \autoref{t:optsharescharacterization}, the two agents receive equal equity shares $\sharesagt{1} = \sharesagt{2} = s/2$ and we obtain $G \sharesvector = s/2$.

Suppose the principal instead allocated $s$ shares among all three agents so that the balanced neighborhood equity condition is satisfied. This requires
$${\sigma}_{1} = \frac{1+G_{13} - G_{23}}{2G_{13}} \cdot c, {\sigma}_{2} = \frac{1+G_{23} - G_{13}}{2G_{23}} \cdot c, {\sigma}_{3} = \frac{G_{23} + G_{13} - 1}{2G_{23}G_{13}} \cdot c.$$
Since $G_{13} + G_{23} > 1$, the three shares ${\sigma}_{1},{\sigma}_{2},{\sigma}_{3}>0$, so we have defined a feasible equity allocation. Combining the expressions for the $\sigma_i$,
$$c = \frac{2G_{13}G_{23}s}{2(G_{23}+G_{13}) - 1 -(G_{13}-G_{23})^{2}} > 0.$$

It remains to check that this allocation gives a higher constant $c$ than the allocation $\sigma_1=\sigma_2=s/2$ and $\sigma_3 =0$ with two active agents. This holds if
    \begin{equation*}
    \begin{split}
        \frac{2G_{13}G_{23}s}{2(G_{23}+G_{13}) - 1 -(G_{13}-G_{23})^{2}} > s/2, \\
        \iff \left(G_{13} + G_{23} - 1\right)^{2} > 0,
    \end{split}
    \end{equation*}
and the bottom inequality holds since $G_{13} + G_{23} > 1$. This shows that there is a unique optimal allocation and the active set at this allocation is $\{1,2,3\}$.

Finally, we prove that $\optsharesagt{1} \geq \optsharesagt{2} \geq \optsharesagt{3}$ by considering both possible active sets. When all agents are active, by \autoref{t:optsharescharacterization}, the optimal shares satisfy
$$\optsharesagt{2} + G_{13}\optsharesagt{3} = \optsharesagt{1} + G_{23}\optsharesagt{3} = G_{13}\optsharesagt{1} + G_{23}\optsharesagt{2},$$
where $\optsharesagt{1},\optsharesagt{2},\optsharesagt{3} > 0$. We have assumed $G_{23}>0$, so $\optsharesagt{1} \geq \optsharesagt{2} \geq \optsharesagt{3}$.

 When only agents $1$ and $2$ are active, by \autoref{t:optsharescharacterization}, the two active agents $i$ and $j$ receive equal shares $\sigma^*_i=\sigma^*_j = \sigma^*$ while agent $3$ receives no shares.
\end{proof}

\begin{proof}[Proof of \autoref{p:compstatedgeweight}]
    Since we are working in the range of values where the optimal active set is all three agents, \autoref{t:optsharescharacterization} implies

$$\frac{\optsharesagt{2}}{\optsharesagt{3}} = \frac{G_{13}\left(G_{23} + 1 - G_{13}\right)}{G_{23} + G_{13} - 1}. $$

Taking the derivative of the ratio of the optimal shares with respect to $G_{23}$, we get that

$$\frac{\partial \left(\frac{\optsharesagt{2}}{\optsharesagt{3}}\right)}{\partial G_{23}} = -\frac{2G_{13}(1-G_{13})}{\left(G_{23} + G_{13} - 1\right)^{2}} < 0.$$

Turning to the relative shares between agent $1$ and $3$, \autoref{t:optsharescharacterization} implies

$$\frac{\optsharesagt{1}}{\optsharesagt{3}} = \frac{G_{23}\left(G_{13} + 1 - G_{23}\right)}{G_{23} + G_{13} - 1}.$$

Taking the derivative of this ratio with respect to $G_{23}$, we get that

$$\frac{\partial \left(\frac{\optsharesagt{1}}{\optsharesagt{3}}\right)}{\partial G_{23}} = \frac{(G_{13}-1)(1+G_{13}-G_{23}) - G_{23}(G_{23} + G_{13}-1)}{\left(G_{23} + G_{13} - 1\right)^{2}} < 0,$$
where the inequality follows from \autoref{p:threeagentoptshareordering}. Finally, we have that

$$\frac{\optsharesagt{1}}{\optsharesagt{2}} = \frac{(1+G_{13} - G_{23})G_{23}}{(1+G_{23} - G_{13})G_{13}}.$$

Taking the derivative with respect to $G_{23}$ and analyzing the quadratic function in $G_{23}$ in the numerator gives the result with threshold  $$g^*(G_{13})= \sqrt{1-G_{13}}\left(\sqrt{2} - \sqrt{1-G_{13}}\right).\eqno \qedhere$$
\end{proof}

\section{Spectral radius maximization}
\label{a:spectral_radius}

We now state and proves the connection to the spectral radius maximization problem discussed in \Cref{s:spectral_radius}. Recall the spectral radius of a matrix $\mathbf{M}$, which we write as $\rho(\mathbf{M})$, is the largest magnitude of an eigenvalue of $\mathbf{M}$.

\begin{proposition}

    \label{p:activesetimpspecmax}
    If $\sharesvector^*$ is optimal for the success probability objective, then $\sharesmatrix^* = \operatorname{diag}(\sharesvector^*)$ maximizes the spectral radius $\rho(\sharesmatrix \network )$ among matrices $\sharesmatrix=\operatorname{diag}(\sharesvector)$ such that $\sharesvector$ is a feasible allocation.

\end{proposition}

\begin{proof}

For the proof, we consider the alternate model (discussed in \Cref{s:objectives}) in which the team has a deterministic output $P(Y)$ and the principal allocates equity shares $\sigma_i P(Y)$ to each agent $i$. We moreover assume that $P(Y) = Y$. Our analysis applies essentially unchanged, with equilibrium actions under shares $\sharesvector$ now given by $$\eqlbactions = [I-\beta\sharesmatrix \network]^{-1} \sharesvector.$$ In particular, the optimal allocations are unchanged and \Cref{p:invariance} continues to apply in this alternate model.

    Suppose $\boldsymbol{\sigma}^{*}$ is optimal for the output objective for some complementarity parameter $\beta_0$ but does not maximize the spectral radius $\rho(\sharesmatrix \network )$. Let $\boldsymbol{\sigma}'$ maximize the spectral radius $\rho(\sharesmatrix \network )$. Then, by assumption, $\rho(\boldsymbol{\Sigma}^{*}\mathbf{G}) < \rho(\boldsymbol{\Sigma}'\mathbf{G})$. Recall that given shares $\sharesvector$, the equilibrium actions are $$\eqlbactions = [I-\beta\sharesmatrix \network]^{-1} \sharesvector.$$Since $\sharesvector$ is the right eigenvector of $\sharesmatrix \network$ with largest eigenvalue, we have $$\eqlbactions = \frac{1}{1-\beta\rho(\sharesmatrix \network)}\sharesvector.$$
    
    Taking the limit of $\beta$ to $1/\rho(\boldsymbol{\Sigma}'\mathbf{G})$ from below, we have $\|\eqlbactions\|_2 \rightarrow \infty$ under the equity allocation $\sharesvector'$ but $\|\eqlbactions\|_2$ remains bounded under the equity allocation $\optsharesvector$. So we can choose $\beta_1$ such that the team performance with shares $\boldsymbol{\sigma}'$ is strictly greater than the team performance with shares $\boldsymbol{\sigma}^{*}$. This implies that  $\boldsymbol{\sigma}^{*}$ does not maximize the output objective with complementarity parameter $\beta_1$. \Cref{p:invariance} implies that $\boldsymbol{\sigma}^{*}$ does not maximize the success probability with complementarity parameter $\beta_0$ either, giving a contradiction. Thus, any solution to the principal's problem maximizes the spectral radius $\rho(\sharesmatrix \network )$. 
\end{proof}
\end{document}

%% file: notation.tex
\newcommand{\production}{Y}
\newcommand{\residualprofits}{V}
\newcommand{\actions}{\mathbf{a}}

\newcommand{\eqlbactions}{\mathbf{a}^{*}}

\newcommand{\actionagt}[1]{a_{#1}}
\newcommand{\network}{\mathbf{G}}
\newcommand{\bonacichendo}{\mathbf{M}}
\newcommand{\subnetwork}{\widetilde{\mathbf{G}}}
\newcommand{\sharesvector}{\boldsymbol{\sigma}}
\newcommand{\optsharesvector}{\boldsymbol{\sigma}^{*}}
\newcommand{\sharesmatrix}{\mathbf{\Sigma}}
\newcommand{\optsharesmatrix}{\mathbf{\Sigma}^{*}}
\newcommand{\sharesagt}[1]{\sigma_{#1}}
\newcommand{\optsharesagt}[1]{\sigma^{*}_{#1}}

\newcommand{\nghbd}[1]{N(#1)}